\documentclass[]{elsarticle}

\usepackage{hyperref}
\usepackage{amsfonts}
\usepackage{amssymb}
\usepackage{amsmath}
\usepackage{multicol}
\usepackage{amsthm}
\usepackage{comment}
\usepackage{placeins}
\usepackage{appendix}
\usepackage{subcaption}
\usepackage{xcolor}
\usepackage{algorithm}
\usepackage{algorithmic} 
\usepackage{hyperref}
\include{pythonlisting}

\journal{Journal}

\begin{document}

\newcommand{\theoremName}{Theorem}
\newcommand{\corollaryName}{Corollary}
\newcommand{\propositionName}{Proposition}
\newcommand{\lemmaName}{Lemma}
\newcommand{\propertyName}{Property}
\newcommand{\definitionName}{Definition}
\newcommand{\axiomName}{Axiom}
\newcommand{\remarkName}{Remark}
\newcommand{\exampleName}{Example}

\newtheorem{theorem}{\theoremName}[section]
\newtheorem{corollary}[theorem]{\corollaryName}
\newtheorem{proposition}[theorem]{\propositionName}
\newtheorem{lemma}[theorem]{\lemmaName}

\newtheorem{property}{\propertyName}[section]
\newtheorem{definition}{\definitionName}[section]
\newtheorem{axiom}{\axiomName}[section]
\newtheorem{remark}{\remarkName}[section]
\newtheorem{example}{\exampleName}[section]

\begin{frontmatter}

\title{Efficient estimation of multiple expectations with the same sample by adaptive importance sampling and control variates}


\author[1,2]{Julien \textsc{Demange-Chryst}\corref{mycorrespondingauthor}}
\cortext[mycorrespondingauthor]{Corresponding author}
\ead{julien.demange-chryst@onera.fr}

\author[2]{François \textsc{Bachoc}}

\ead{francois.bachoc@math.univ-toulouse.fr}

\author[1]{Jérôme \textsc{Morio}}
\ead{jerome.morio@onera.fr}

\address[1]{ONERA/DTIS, Université de Toulouse, F-31055 Toulouse, France}
\address[2]{Institut de Mathématiques de Toulouse, UMR5219 CNRS, 31062 Toulouse, France}

\begin{abstract}
Some classical uncertainty quantification problems require the estimation of multiple expectations. Estimating all of them accurately is crucial and can have a major impact on the analysis to perform, and standard existing Monte Carlo methods can be costly to do so. We propose here a new procedure based on importance sampling and control variates for estimating more efficiently multiple expectations with the same sample. We first show that there exists a family of optimal estimators combining both importance sampling and control variates, which however cannot be used in practice because they require the knowledge of the values of the expectations to estimate. Motivated by the form of these optimal estimators and some interesting properties, we therefore propose an adaptive algorithm. The general idea is to adaptively update the parameters of the estimators for approaching the optimal ones. We suggest then a quantitative stopping criterion that exploits the trade-off between approaching these optimal parameters and having a sufficient budget left. This left budget is then used to draw a new independent sample from the final sampling distribution, allowing to get unbiased estimators of the expectations. We show how to apply our procedure to sensitivity analysis, by estimating Sobol' indices and quantifying the impact of the input distributions. Finally, realistic test cases show the practical interest of the proposed algorithm, and its significant improvement over estimating the expectations separately.
\end{abstract}


\end{frontmatter}

\section{Introduction}


Some classical uncertainty quantification problems require the estimation of multiple expectations, and estimating all of them accurately is crucial. The generalized method of moments \cite{hansen1982large}, which is massively used in finance for example \cite{jagannathan2002generalized}, is a common illustration of a such problem. Another classical illustration of this problematic is global sensitivity analysis \cite{saltelli2004sensitivity}, which aims at studying the impact of the input variables on the output behaviour of a computer model. Performing a such study consists in estimating some sensitivity indices associated to each input variable, such as the Sobol' indices \cite{sobol1993sensitivity} or the Shapley effects \cite{owen2014sobol} for example, and requires in each case the estimation of multiple expectations.




The usual quadrature methods \cite{davis2007methods} tend not to be appropriate in these uncertainty quantification contexts, as the expectations then involve a numerical model which computational cost is usually high (from several minutes to several days CPU), and which number of input variables is not small. Standard existing Monte Carlo methods \cite{rubinstein2016simulation} for estimating multiple expectations consist in drawing a unique sample according to a given input distribution and to estimate all of them with it. However, this sample can be ill-suited for estimating accurately some of the expectations, so having accurate estimations of all of them can be costly with this method. As a consequence, the resulting error can have a major impact on the final goal of the analysis, as illustrated in our numerical experiments in Section~\ref{sec:numerical_results}. Importance sampling \cite{Kahn1951SplittingParticleTransmission} and control variates \cite{nelson1987control} are two well-known and deeply investigated variance-reduction techniques for improving the estimation of a single expectation. However, to the best of our knowledge, these methods have not been adapted for jointly estimating multiple expectations with the same sample.

 
In this article, we first propose a criterion to quantify the quality of the common estimation of multiple expectations with the same sample. We show then that there exists a family of optimal estimators combining both importance sampling and control variates. However, these optimal estimators cannot be used in practice because they require the knowledge of the values of the expectations to estimate. Motivated by the form of these optimal estimators and some interesting properties \cite{owen2000safe,he2014optimal}, we therefore propose an adaptive algorithm called ME-aISCV combining both importance sampling and control variates for estimating multiple expectations with the same sample. Not only can we address different functions across the expectations, but also different input distributions. In the same way as other adaptive algorithms \cite{cornuet2012adaptive,marin2012consistency}, the general idea is to sequentially update the parameters of the estimators for approaching the optimal ones until a stopping criterion is reached. We suggest a quantitative stopping criterion that exploits the trade-off between approaching these optimal parameters and having a sufficient budget left. At last, the left budget is used to draw a new independent sample according to the final sampling distribution which allows to get unbiased estimators of the expectations to estimate.


The remainder of this paper is organized as follows. First, Section \ref{sec:review} formally presents the problem and provides a review on importance sampling and control variates. Then, Section \ref{sec:new_algorithm} introduces and describes the proposed ME-aISCV algorithm for estimating multiple expectations with the same sample. In addition, Section \ref{sec:numerical_results} illustrates the practical interest of this new algorithm on the estimation of several moments of the standard Gaussian distribution. It then shows that the ME-aISCV algorithm can be applied to the estimation of first order Sobol' indices and to sensitivity analysis w.r.t. parameters of the input distribution. Both applications are illustrated on a real structural engineering example: the cantilever beam problem. In all cases, the improvement of our methodology over estimating the expectations separately is significant. Finally, Section \ref{sec:conclusion} concludes the present article and gives future research perspectives stemming from it.

\section{Exposition of the problem and review on variance-reduction methods}
\label{sec:review}

In this section, we first expose the problem of estimating multiple expectations with the same sample and we recall the main principles of importance sampling and control variates to address it.

First of all, let us begin by introducing the notations that will be used throughout the paper. For any probability density $h$ from the input domain $\mathbb{X} = \bigotimes_{i=1}^d \mathbb{X}_i \subseteq \mathbb{R}^d$ to $\mathbb{R}_+$, we let $\mathbb{E}_h$ and $\mathbb{V}_h$ denote respectively the expectation and the variance operators of a random variable distributed according to $h$. Then, for $J\geq2$, we consider a family of non-negative functions $\left(\phi_j\right)_{j\in[\![1,J]\!]}$ from $\mathbb{X}$ to $\mathbb{R}_+$. Moreover, for any $j\in[\![1,J]\!]$, the random input vector $\mathbf{X} = \left(X_1,\dots,X_d\right)$ of the function $\phi_j$ on $\mathbb{X}$ follows the distribution of joint PDF $f_{j}$. No regularity assumption on the functions is required, but the random output of each function is supposed to be integrable, i.e. $\mathbb{E}_{f_j}\left(\phi_j\left(\mathbf{X}\right)\right) < + \infty$.

\subsection{Estimating multiple expectations with the same sample}
\label{ssec:presentation_problem}

As discussed and motivated in the introduction, the main goal of this article is to efficiently estimate multiple expectations while minimising the number of calls to the functions $\left(\phi_j\right)_{j\in[\![1,J]\!]}$ using a unique $N$-sample. More precisely, the family of expectations to estimate is $\left(I_j = \mathbb{E}_{f_j}\left[\phi_j\left(\mathbf{X}\right)\right]\right)_{j\in[\![1,J]\!]}$, the $N$-sample is $\left(\mathbf{X}^{(n)}\right)_{n\in[\![1,N]\!]}$ and it is drawn from a distribution of PDF $g$.

In practice, two specific cases can occur: \begin{itemize}
    \item Case 1: estimating the expectation of $J$ different functions under the same input distribution, or formally $\forall i,j\in [\![1,J]\!]$, $i\neq j \Longrightarrow \phi_i \neq \phi_j$ and $\forall j\in [\![1,J]\!]$, $f_j = f$, see Section \ref{ssec:sobol_estimation} for a numerical example,
    \item Case 2: estimating the expectation of the same function $\phi$ under $J$ different input distributions, or formally $\forall j\in [\![1,J]\!]$, $\phi_j = \phi$ and $\forall i,j\in [\![1,J]\!]$, $i\neq j \Longrightarrow f_i \neq f_j$, see Section \ref{ssec:sensitivity_wrt_input_params} for a numerical example.
\end{itemize}

The quality of the estimation of one expectation can be evaluated with the variance for unbiased estimators. When estimating $J$ expectations, a natural criterion is the weighted sum of the individual variance of each estimator, which is briefly mentioned in \cite{he2014optimal}. To define this criterion, let us consider a family of positive weights $\left(w_j\right)_{j\in[\![1,J]\!]}\in\mathbb{R}_+^J$. Then, for any $j\in[\![1,J]\!]$, let us denote $\widehat{I}_j$ an estimator of the expectation $I_j$ such that all the estimators  $\widehat{I}_1,\dots,\widehat{I}_J$ are based on the same $N$-sample distributed according to $g$. The criterion we want to minimize is:
\begin{equation}\label{eq:criterion}
    \sum_{j=1}^J w_j\mathbb{V}_g\left(\widehat{I}_j\right).
\end{equation} The positive weights $\left(w_j\right)_{j\in[\![1,J]\!]}$ can be used to adjust the importance given to each expectation to estimate.

\subsection{Importance sampling}
\label{ssec:importance_sampling}

\subsubsection{General presentation}
\label{ssec:general_presentation}

\textit{Importance sampling} (IS) is a very usual variance-reduction technique which was introduced in \cite{Kahn1951SplittingParticleTransmission}. In the case of the estimation of an expectation $I = \mathbb{E}_{f}\left(\phi\left(\mathbf{X}\right)\right)$, it consists in rewriting the expectation according to an auxiliary density $g : \mathbb{X} \longrightarrow \mathbb{R}_+$  as $\mathbb{E}_{g}\left(\phi\left(\mathbf{X}\right)w^g\left(\mathbf{X}\right)\right)$, where $w^g\left(\mathbf{x}\right) = f(\mathbf{x})/g(\mathbf{x})$ is the \textit{likelihood ratio}. To get an unbiased estimate, the support of $g$ must contain the support of $\mathbf{x}\in\mathbb{X} \mapsto \phi\left(\mathbf{x}\right)f(\mathbf{x)}$. The corresponding estimator is then given by: \begin{equation}\label{ptisest}\widehat{I}_{g,N}^{\text{IS}} = \frac{1}{N}\sum_{n=1}^N \phi\left(\mathbf{X}^{(n)}\right)w^g\left(\mathbf{X}^{(n)}\right),\end{equation} where $\left(\mathbf{X}^{(n)}\right)_{n\in [\![1,N]\!]}$ is an i.i.d. sample distributed according to the IS auxiliary distribution $g$. It is consistent and unbiased, and it has zero-variance if and only if $g = g^*$ with $\forall \mathbf{x} \in \mathbb{X}$, $g^*\left(\mathbf{x}\right) \propto \phi\left(\mathbf{x}\right)f(\mathbf{x)}$ \cite{bucklew2004introduction} on the condition that $\phi$ is non-negative. This optimal density cannot be used in practice because the normalizing constant is $I$, which is the quantity to estimate, but many techniques exist to approach $g^*$ by a near-optimal auxiliary density: non-parametric methods \cite{zhang1996nonparametric} or parametric methods such that the cross-entropy method \cite{de2005tutorial,rubinstein2013cross}.

\subsubsection{The cross-entropy method}

In this article, we will seek an approximation of $g^*$ in parametric families of distribution $\mathcal{D}_{\Lambda} = \left\lbrace g_{\boldsymbol{\lambda}} ; \boldsymbol{\lambda}\in\Lambda\right\rbrace$. As a first option, one could aim for the parameter $\boldsymbol{\lambda}^*_{\mathbb{V}}$ which minimizes the variance of the estimator: \begin{equation}\label{eq:variance_minimization}
    \boldsymbol{\lambda}^*_{\mathbb{V}} = \underset{\boldsymbol{\lambda}\in\Lambda}{\mathrm{argmin}} \ \mathbb{V}_{g_{\boldsymbol{\lambda}}}\left(\widehat{I}_{g_{\boldsymbol{\lambda}},N}^{\text{IS}}\right).
\end{equation} However, this optimisation problem is not convex w.r.t. $\boldsymbol{\lambda}\in\Lambda$, does not have an analytical solution and needs to be solved numerically \cite{rubinstein2016simulation}, even for classical families $\mathcal{D}_{\Lambda}$ (like the Gaussian family defined below), which can be extremely costly. Therefore, one typically prefers to use the cross-entropy method. It consists in minimizing the Kullback-Leibler divergence \cite{kullback1951information} between $g^*$ and $g_{\boldsymbol{\lambda}}$ for $\boldsymbol{\lambda}\in\Lambda$ in order to find the best representative of $g^*$ in $\mathcal{D}_{\Lambda}$. The Kullback-Leibler divergence between two distributions of PDF $g_1$ and $g_2$ is given by: \begin{equation}
    D_{\text{KL}}\left(g_1,g_2\right) = \mathbb{E}_{g_1}\left(\log\left(\dfrac{g_1\left(\mathbf{X}\right)}{g_2\left(\mathbf{X}\right)}\right)\right) = \int_{\mathbb{X}}\log\left(\dfrac{g_1\left(\mathbf{x}\right)}{g_2\left(\mathbf{x}\right)}\right)g_1\left(\mathbf{x}\right)d\mathbf{x}.
\end{equation} The quantity $D_{\text{KL}}\left(g_1,g_2\right)$ is always non-negative and is zero if and only if $g_1 = g_2$ almost everywhere. It measures the gap between two distributions, even if it is not a distance because it is not symmetric. The cross-entropy method consists then in finding the solution $\boldsymbol{\lambda}^*$ of the optimization problem: \begin{equation}\label{eq:cross_entropy_min}
    \boldsymbol{\lambda}^* = \underset{\boldsymbol{\lambda}\in\Lambda}{\mathrm{argmin}} \  D_{\text{KL}}\left(g^*,g_{\boldsymbol{\lambda}}\right).
\end{equation} Under this form, this optimization cannot be solved because it depends explicitly on $g^*$ which is unknown. However, it can be shown \cite{de2005tutorial} that the optimization problem in \eqref{eq:cross_entropy_min} is equivalent to solve: \begin{equation}\label{eq:cross-entropy_max}
    \boldsymbol{\lambda}^* = \underset{\boldsymbol{\lambda}\in\Lambda}{\mathrm{argmax}} \  \mathbb{E}_f\left[\log\left(g_{\boldsymbol{\lambda}}\left(\mathbf{X}\right)\right)\phi\left(\mathbf{X}\right)\right].
\end{equation} In opposition to the variance-minimization problem in \eqref{eq:variance_minimization}, the cross-entropy problem in \eqref{eq:cross-entropy_max} is generally concave and differentiable w.r.t. $\boldsymbol{\lambda}\in\Lambda$ \cite{rubinstein2013cross}. Another significant advantage of the problem in \eqref{eq:cross-entropy_max} is that it has an analytical solution when $\mathcal{D}_{\Lambda}$ belongs to the exponential family of distributions \cite{rubinstein2013cross}.

\subsubsection{Classical families of distributions for the auxiliary distribution}

One of the most famous family of distributions is the Gaussian family $\mathcal{D}_{\text{Gauss}} = \left\lbrace g_{\boldsymbol{m},\boldsymbol{\Sigma}} ;\boldsymbol{m}\in\mathbb{R}^d,\boldsymbol{\Sigma}\in \mathcal{S}_d^+\right\rbrace$, which belongs to the exponential family. Each Gaussian distribution is fully determined by $\boldsymbol{\lambda} = \left(\boldsymbol{m},\boldsymbol{\Sigma}\right)$, with $\boldsymbol{m}\in\mathbb{R}^d$ the mean vector and $\boldsymbol{\Sigma}\in \mathcal{S}_d^+$ the covariance matrix, where $\mathcal{S}_d^+$ denotes the set of symmetric positive-definite real-valued matrices of size $d\times d$. This family is well-suited when $g^*$ is unimodal. Since $\mathcal{D}_{\text{Gauss}}$ belongs to the exponential family, the cross-entropy problem in \eqref{eq:cross-entropy_max} has an analytical solution and it is given by $\boldsymbol{\lambda}^* = \left(\boldsymbol{m}^*,\boldsymbol{\Sigma}^*\right)$: \begin{equation}
    \boldsymbol{m}^* = \dfrac{\mathbb{E}_f\left[\phi\left(\mathbf{X}\right)\mathbf{X}\right]}{\mathbb{E}_f\left[\phi\left(\mathbf{X}\right)\right]} \mbox{ and } \boldsymbol{\Sigma}^* = \dfrac{\mathbb{E}_f\left[\phi\left(\mathbf{X}\right)\left(\mathbf{X}-\boldsymbol{m}^*\right)\left(\mathbf{X}-\boldsymbol{m}^*\right)^\top\right]}{\mathbb{E}_f\left[\phi\left(\mathbf{X}\right)\right]}. 
\end{equation} In practice, these optimal parameters are estimated with a sample, which is called the stochastic counterpart \cite{de2005tutorial}.

The optimal density $g^*$ can also be multimodal. In that case, a well-suited family of distributions is the Gaussian mixture family \cite{kurtz2013cross}. Let us first define, for any $K\geq1$, the set of convex combinations of size $K$:  \begin{equation}
   S_K = \left\lbrace \left(\alpha_j\right)_{j\in[\![1,K]\!]} ; \sum_{k=1}^K \alpha_k = 1 \mbox{ and } \forall k \in[\![1,K]\!], \alpha_k\geq0\right\rbrace.
\end{equation} Then, the Gaussian mixture family with $K\geq1$ components is given by $\mathcal{D}_{\text{Mix}}^{(K)} = \left\lbrace \sum_{k=1}^K \alpha_k g_{\boldsymbol{m}_k,\boldsymbol{\Sigma}_k} ;\left(\boldsymbol{m}_k\right)_{k\in[\![1,K]\!]}\in\left(\mathbb{R}^d\right)^K,\left(\boldsymbol{\Sigma}_k\right)_{k\in[\![1,K]\!]}\in \left(\mathcal{S}_d^+\right)^{K},\left(\alpha_k\right)_{k\in[\![1,K]\!]}\in S_K\right\rbrace$. The Gaussian mixture family does not belong to the exponential family, but since solving the cross-entropy problem is equivalent to obtaining the maximum likelihood estimate of the parameters \cite{rubinstein2016simulation}, it is possible to use the Expectation-Maximisation algorithm \cite{dempster1977maximum} to estimate them efficiently thanks to the procedure described in \cite{chen2010demystified,geyer2019cross}.

\subsection{Control variates}
\label{ssec:control_variates}

\subsubsection{General presentation}

\textit{Control variates} (CV) is another variance-reduction technique \cite{nelson1987control}. It consists in exploiting known values of some integrals of control functions in order to improve the quality of the estimation of an expectation. CV has been first defined as a straightforward extension of the Monte Carlo estimate of the expectation \cite{nelson1987control,nelson1990control}, but it can be paired with IS \cite{owen1999adaptive,owen2000safe}. For the sake of conciseness, we will describe CV with only one control function, but it can be easily generalized to the case of multiple control functions. More precisely, let us consider a control function $h : \mathbb{X} \longrightarrow \mathbb{R}$ such that $\int_{\mathbb{X}}h\left(\mathbf{x}\right)d\mathbf{x} = \theta \in \mathbb{R}$ is known, and a real value $\beta \in \mathbb{R}$ called control parameter. Then, \begin{equation}\label{eq:control_variates_estimator}
    \widehat{I}_{g,h,\beta,N}^{\text{CV}} = \dfrac{1}{N}\sum_{n=1}^N\dfrac{\phi\left(\mathbf{X}^{(n)}\right)f\left(\mathbf{X}^{(n)}\right)-\beta h\left(\mathbf{X}^{(n)}\right)}{g\left(\mathbf{X}^{(n)}\right)} + \beta\theta,
\end{equation} where $\left(\mathbf{X}^{(n)}\right)_{[\![1,N]\!]}$ is an i.i.d. sample drawn according to $g$, is an unbiased estimator with CV and IS of $I$. Its variance is then given by: \begin{align}
    N\mathbb{V}_g&\left(\widehat{I}_{g,h,\beta,N}^{\text{CV}}\right) = \mathbb{V}_g\left(\dfrac{\phi\left(\mathbf{X}\right)f\left(\mathbf{X}\right)-\beta h\left(\mathbf{X}\right)}{g\left(\mathbf{X}\right)}\right) \label{eq:variance_control_variates}\\
    &= \mathbb{V}_g\left(\dfrac{\phi\left(\mathbf{X}\right)f\left(\mathbf{X}\right)}{g\left(\mathbf{X}\right)}\right) - 2\beta\mathrm{Cov}_g\left(\dfrac{\phi\left(\mathbf{X}\right)f\left(\mathbf{X}\right)}{g\left(\mathbf{X}\right)},\dfrac{h\left(\mathbf{X}\right)}{g\left(\mathbf{X}\right)}\right) + \beta^2\mathbb{V}_g\left(\dfrac{h\left(\mathbf{X}\right)}{g\left(\mathbf{X}\right)}\right). \label{eq:variance_control_variates_dev}
\end{align} By minimising Equation \eqref{eq:variance_control_variates_dev} according to the real parameter $\beta$, it can be shown that the optimal value of $\beta$ is: \begin{equation}\label{eq:beta_star_cv}
\beta^* = \mathbb{V}_g\left(\dfrac{h\left(\mathbf{X}\right)}{g\left(\mathbf{X}\right)}\right)^{-1}\mathrm{Cov}_g\left(\dfrac{\phi\left(\mathbf{X}\right)f\left(\mathbf{X}\right)}{g\left(\mathbf{X}\right)},\dfrac{h\left(\mathbf{X}\right)}{g\left(\mathbf{X}\right)}\right).\end{equation} This optimal value $\beta^*$ satisfies $\mathbb{V}_g\left(\widehat{I}_{g,h,\beta^*,N}^{\text{CV}}\right) \leq \mathbb{V}_g\left(\widehat{I}_{g,N}^{\text{IS}}\right)$, which means that it is possible to improve the quality of the estimation of $I$ with CV if the parameter $\beta$ is chosen carefully. In practice, the optimal parameter $\beta^*$ is estimated either directly through Equation \eqref{eq:beta_star_cv} \cite{glynn2002some} or by a least square regression by minimising Equation \eqref{eq:variance_control_variates} \cite{owen1999adaptive,leluc2021control}.

At last, note that if we use the same sample to compute an estimator $\widehat{\beta}$ of $\beta^*$ and the expectation $I$ by plugging $\widehat{\beta}$ in \eqref{eq:control_variates_estimator}, the estimator $\widehat{I}_{g,h,\widehat{\beta},N}^{\text{CV}}$ is biased. However, this bias can be eliminated if we use two different samples to compute $\widehat{\beta}$ and $\widehat{I}_{g,h,\widehat{\beta},N}^{\text{CV}}$.

\subsubsection{Mixture importance sampling with control variates}

A mixture of $K\geq1$ distributions $g_1,\dots,g_K$ is a distribution of the form $g_{\boldsymbol{\alpha}} = \sum_{k=1}^K \alpha_kg_k$, where the sequence of real numbers $\boldsymbol{\alpha} = \left(\alpha_k\right)_{k\in[\![1,K]\!]}$ belongs to $S_K$. For example, an element of the family $\mathcal{D}_{\text{Mix}}^{(K)}$ is a mixture of $K$ Gaussian distributions. The use of mixture distributions as IS auxiliary distributions without and with CV can be beneficial in order to deal with multimodal problems and satisfies as well some interesting properties \cite{owen2000safe,owen2013mcbook,he2014optimal}, some of which are described below.

Assume that for all $k\in[\![1,K]\!]$, the support of $g_k$ contains the support of $\mathbf{x}\in\mathbb{X} \mapsto \phi\left(\mathbf{x}\right)f(\mathbf{x})$. This assumption implies that for all $\boldsymbol{\alpha}\in S_K$, for all $\beta \in \mathbb{R}$ and for all $k\in[\![1,K]\!]$, the support of the mixture distribution $g_{\boldsymbol{\alpha}}$ contains the support of $\mathbf{x}\in\mathbb{X} \mapsto \phi\left(\mathbf{x}\right)f\left(\mathbf{x}\right) - \beta g_k\left(\mathbf{x}\right)$. Then, the authors of \cite{owen2000safe,he2014optimal} proved the following theorem. \begin{theorem}\label{thm:variance_cv_is} For any $k\in [\![1,K]\!]$ and $\boldsymbol{\alpha}\in S_K$, we have:
\begin{equation}
    N\mathbb{V}_{g_{\boldsymbol{\alpha}}}\left(\widehat{I}_{g_{\boldsymbol{\alpha}},g_k,\beta^*,N}^{\text{CV}}\right) = \mathbb{V}_{g_{\boldsymbol{\alpha}}}\left(\dfrac{\phi\left(\mathbf{X}\right)f\left(\mathbf{X}\right)-\beta^* g_k\left(\mathbf{X}\right)}{g_{\boldsymbol{\alpha}}\left(\mathbf{X}\right)}\right) \leq \alpha_k^{-1}\mathbb{V}_{g_k}\left(\dfrac{\phi\left(\mathbf{X}\right)f\left(\mathbf{X}\right)}{g_k\left(\mathbf{X}\right)}\right).
\end{equation}
\end{theorem} This theorem ensures that if one component $g_{k_0}$ of the mixture distribution $g_{\boldsymbol{\alpha}} = \sum_{k=1}^K \alpha_kg_k$ is well-suited to the problem of estimating the expectation $I = \mathbb{E}_f\left(\phi\left(\mathbf{X}\right)\right)$, then the variance of the estimator $\widehat{I}_{g_{\boldsymbol{\alpha}},g_{k_0},\beta^*,N}^{\text{CV}}$ using $g_{k_0}$ as control function would be small.

Moreover, the choice of the coefficients $\boldsymbol{\alpha}\in S_K$ of the mixture $g_{\boldsymbol{\alpha}}$ can have a major impact on the variance of the CV estimator. The authors of \cite{owen2000safe,he2014optimal} proved as well the following theorem. \begin{theorem} \label{thm:convex_optim_single_expectation}
For any $\beta\in \mathbb{R}$ and $k\in[\![1,K]\!]$, the optimisation problem
\begin{equation}
    \boldsymbol{\alpha}^* = \underset{\boldsymbol{\alpha}\in S_K}{\mathrm{argmin}} \  \mathbb{V}_{g_{\boldsymbol{\alpha}}}\left(\dfrac{\phi\left(\mathbf{X}\right)f\left(\mathbf{X}\right)-\beta g_k\left(\mathbf{X}\right)}{g_{\boldsymbol{\alpha}}\left(\mathbf{X}\right)}\right)
\end{equation} is convex on $S_K$.
\end{theorem} This theorem ensures then that simple optimisation algorithms can be performed in order to find a sequence of real coefficients $\boldsymbol{\alpha}\in S_K$ which gives a small variance for the IS-CV estimator.

\section{New adaptive algorithm for estimating multiple expectations with the same sample}
\label{sec:new_algorithm}

In this section, we first provide the theoretical motivations leading to a new procedure for estimating $J$ expectations with a unique $N$-sample which minimizes the criterion in Equation \eqref{eq:criterion}. Second, we describe more precisely the proposed ME-aISCV algorithm itself.

Recall that the optimal IS auxiliary distribution for estimating an expectation $I = \mathbb{E}_f\left(\phi\left(\mathbf{X}\right)\right)$ is given for all $\mathbf{x}\in\mathbb{X}$ by $g^*\left(\mathbf{x}\right) = I^{-1}\phi\left(\mathbf{x}\right)f\left(\mathbf{x}\right)$. For $j\in[\![1,J]\!]$, let us then denote $g_j^*$ the optimal IS auxiliary distribution for estimating $I_j = \mathbb{E}_{f_j}\left(\phi_j\left(\mathbf{X}\right)\right)$.

\subsection{Theoretical motivation}
\label{ssec:theoretical_motivations}

Let us begin with the following proposition. \begin{proposition}
For any IS auxiliary distribution $g$ and any i.i.d. sample $\left(\mathbf{X}^{(n)}\right)_{n\in[\![1,N]\!]}$ drawn according to $g$, and for any $j\in[\![1,J]\!]$, the estimator \begin{equation}\label{eq:optimal_cv_estimators}
    \widehat{I}_{g,g_j^*,I_j,N}^{\text{CV}} = \dfrac{1}{N}\sum_{n=1}^N\dfrac{\phi_j\left(\mathbf{X}^{(n)}\right)f_j\left(\mathbf{X}^{(n)}\right)-I_j g_j^*\left(\mathbf{X}^{(n)}\right)}{g\left(\mathbf{X}^{(n)}\right)} + I_j
\end{equation} is an unbiased zero-variance estimator of the expectation $I_j = \mathbb{E}_{f_j}\left(\phi_j\left(\mathbf{X}\right)\right)$.
\end{proposition} 
\begin{proof}
By plugging the expressions of $g_j^*$ in the estimator $\widehat{I}_{g,g_j^*,I_j,N}^{\text{CV}}$, a simple computation leads to $\widehat{I}_{g,g_j^*,I_jN}^{\text{CV}} = I_j$. Equivalently, $\mathbb{E}_g\left(\widehat{I}_{g,g_j^*,I_j,N}^{\text{CV}}\right) = I_j$ and $\mathbb{V}_g\left(\widehat{I}_{g,g_j^*,I_j,N}^{\text{CV}}\right) = 0$.
\end{proof}

Note that $I_j$ corresponds in this case to the optimal value of the control parameter $\beta \in \mathbb{R}$ given in Equation \eqref{eq:beta_star_cv}. This proposition implies that for any sequence $\left(w_j\right)_{j\in[\![1,J]\!]}\in \mathbb{R}_+$, we have: \begin{equation}
    \sum_{j=1}^J w_j\mathbb{V}_g\left(\widehat{I}_{g,g_j^*,I_j,N}^{\text{CV}}\right) = 0.
\end{equation} This result is very interesting because it shows that the use of CV allows to make the criterion to minimise in Equation \eqref{eq:criterion} equal to 0 with any auxiliary sampling distribution. Nevertheless, the estimators in Equation \eqref{eq:optimal_cv_estimators} cannot be used in practice because they require the knowledge of the values of $\left(I_j\right)_{j\in[\![1,J]\!]}$, which are the quantities to estimate.

To overcome this problem, in the same way as in the classical IS framework presented in Section \ref{ssec:importance_sampling}, it is possible to approach these optimal IS distributions $\left(g_j^*\right)_{j\in[\![1,J]\!]}$ by auxiliary distributions $\left(g_{\boldsymbol{\lambda}_j}\right)_{j\in[\![1,J]\!]}$ lying in a parametric family of distributions $\mathcal{D}_{\Lambda}$. We can then plug them in the expression of the estimators in Equation \eqref{eq:optimal_cv_estimators}. The modification of the control functions from $\left(g_j^*\right)_{j\in[\![1,J]\!]}$ to $\left(g_{\boldsymbol{\lambda}_j}\right)_{j\in[\![1,J]\!]}$ implies that the optimal values of the control parameters $\left(\beta_j\right)_{j\in[\![1,j]\!]}$ are no longer equal to the expectations $\left(I_j\right)_{j\in[\![1,J]\!]}$. It is then necessary to estimate these new optimal parameters with some estimators $\left(\widehat{\beta}_j\right)_{j\in[\![1,J]\!]}$ of the expression in Equation \eqref{eq:beta_star_cv}. 

Moreover, the distribution $g_{\boldsymbol{\lambda}_j}$ is usually well-suited to estimate the expectation $I_j$ by IS. Then, Theorem \ref{thm:variance_cv_is} motivates us to consider a mixture $g_{\boldsymbol{\alpha}} = \sum_{j=1}^J \alpha_j g_{\boldsymbol{\lambda}_j}$ as the IS auxiliary sampling distribution. Indeed, since this distribution is a mixture of the $\left(g_{\boldsymbol{\lambda}_j}\right)_{j\in[\![1,J]\!]}$, it is possible to apply Theorem \ref{thm:variance_cv_is} to each estimator $\widehat{I}_{g_{\boldsymbol{\alpha}},g_{\boldsymbol{\lambda}_j},\beta_j^*,N}^{\text{CV}}$, with $\beta_j^*$ the optimal control parameter associated to this problem: \begin{equation}\label{eq:upper_bound_variance_cv}
    N\mathbb{V}_{g_{\boldsymbol{\alpha}}}\left(\widehat{I}_{g_{\boldsymbol{\alpha}},g_{\boldsymbol{\lambda}_j},\beta_j^*,N}^{\text{CV}}\right) \leq \alpha_j^{-1}\mathbb{V}_{g_{\boldsymbol{\lambda}_j}}\left(\dfrac{\phi_j\left(\mathbf{X}\right)f_j\left(\mathbf{X}\right)}{g_{\boldsymbol{\lambda}_j}\left(\mathbf{X}\right)}\right).
\end{equation} This result gives thus an interesting upper bound for the variance of each estimator $\widehat{I}_{g_{\boldsymbol{\alpha}},g_{\boldsymbol{\lambda}_j},\beta_j^*,N}^{\text{CV}}$ for $j\in[\![1,J]\!]$, and thus an upper bound of the criterion to minimize in \eqref{eq:criterion} by summing these upper bounds.

Equation \eqref{eq:upper_bound_variance_cv} highlights as well the importance of the choice of the weights $\boldsymbol{\alpha} = \left(\alpha_j\right)_{j\in[\![1,J]\!]}\in S_J$ of the mixture. Indeed, for $j\in [\![1,J]\!]$, if $\alpha_j\ll 1$ and $\mathbb{V}_{g_{\boldsymbol{\lambda}_j}}\left(\phi_j\left(\mathbf{X}\right)f_j\left(\mathbf{X}\right)\left/g_{\boldsymbol{\lambda}_j}\left(\mathbf{X}\right)\right.\right)$ is large, then the upper bound of the variance of $\widehat{I}_{g_{\boldsymbol{\alpha}},g_{\boldsymbol{\lambda}_j},\widehat{\beta}_j,N}^{\text{CV}}$ will be bad. The intuition given by Equation \eqref{eq:upper_bound_variance_cv} is that high values of $w_j\mathbb{V}_{g_{\boldsymbol{\lambda}_j}}\left(\phi_j\left(\mathbf{X}\right)f_j\left(\mathbf{X}\right)\left/g_{\boldsymbol{\lambda}_j}\left(\mathbf{X}\right)\right.\right)$ must be associated to high values of $\alpha_j$, and the other way around. It is thus beneficial to optimize the choice of $\boldsymbol{\alpha}$, which is facilitated  by the following extension of Theorem \ref{thm:convex_optim_single_expectation} to the case of multiple expectations.
\begin{theorem}
For any $\left(\beta_j\right)_{j\in[\![1,J]\!]}\in \mathbb{R}^J$ and any family of positive weights $\left(w_j\right)_{j\in[\![1,J]\!]}\in \mathbb{R}_+^J$, the optimisation problem
\begin{equation}\label{eq:optimiation_alpha_k}
    \boldsymbol{\alpha}^* = \underset{\boldsymbol{\alpha}\in S_J}{\mathrm{argmin}} \ \sum_{j=1}^J w_j \mathbb{V}_{g_{\boldsymbol{\alpha}}}\left(\dfrac{\phi_j\left(\mathbf{X}\right)f_j\left(\mathbf{X}\right)-\beta_j g_{\boldsymbol{\lambda}_j}\left(\mathbf{X}\right)}{g_{\boldsymbol{\alpha}}\left(\mathbf{X}\right)}\right)
\end{equation} is convex on $S_J$.
\end{theorem} 
\begin{proof}
Theorem \ref{thm:convex_optim_single_expectation} ensures that each individual term of the sum in Equation \eqref{eq:optimiation_alpha_k} is convex on $S_J$ w.r.t. $\boldsymbol{\alpha}$. Therefore, since it is a linear combination with positive weights of convex functions, this optimisation problem is also convex on $S_J$ w.r.t. $\boldsymbol{\alpha}$.
\end{proof} In the same way as in Section \ref{ssec:control_variates}, this theorem ensures then that simple optimisation algorithms can be performed in order to find a sequence of coefficients $\boldsymbol{\alpha}\in S_J$ which reduces the criterion to minimize.

\subsection{Presentation of the algorithm}

\subsubsection{Summary and input parameters}
We propose here a new adaptive algorithm called ME-aISCV to estimate $J$ expectations with the same $N$-sample. In the same way as other adaptive IS algorithms \cite{cornuet2012adaptive,marin2012consistency,de2005tutorial,rubinstein2013cross}, the general idea is to adaptively update the IS auxiliary sampling distributions $\left(g_{\boldsymbol{\lambda}_j}\right)_{j\in[\![1,J]\!]}$, the sampling distribution $g_{\boldsymbol{\alpha}}$ as well as the control parameters $\left(\beta_j\right)_{j\in[\![1,J]\!]}$ until a stopping criterion is reached. Then, a new independent sample drawn according to the final sampling distribution allows to get unbiased estimators by IS and CV of the $J$ expectations.

Let us describe more precisely the ME-aISCV algorithm. As input parameters, it requires the family of functions $\left(\phi_j\right)_{j\in[\![1,J]\!]}$ as well as the corresponding family of input distributions $\left(f_j\right)_{j\in[\![1,J]\!]}$. It requires also the weights $\left(w_j\right)_{j\in[\![1,J]\!]}$, a maximal number of calls allowed to the functions $N_{max}\in\mathbb{N}^*$ and a sequence $\left(N_k\right)_{k\in\mathbb{N}}\in\left(\mathbb{N}^*\right)^{\mathbb{N}}$ corresponding to the number of points to draw at each iteration of the algorithm.

\subsubsection{Initialization}
First, during the initialisation step ($k=0$), an initial $N_0$-sample $\left(\mathbf{X}^{(0,n)}\right)_{n\in[\![1,N_0]\!]}$ is drawn according to an initial sampling distribution $h_0$. This initial sample allows to compute first estimations $\widehat{I}_j^{(0)}$ of the expectations as well as to estimate the new parameters at each iteration of the algorithm. Natural choices for $h_0$ can be either the unweighted mixture $J^{-1}\sum_{j=1}^Jf_j$ or the weighted mixture $\left(\sum_{j=1}^Jw_j\right)^{-1}\sum_{j=1}^Jw_jf_j$. Note that if we are in Case 1 (in Section \ref{ssec:presentation_problem}), i.e. for all $i \in [\![1,J]\!]$ we have $f_i = f$, then $h_0$ is equal to $f$. Then, for $j\in[\![1,J]\!]$, we set $\alpha_j^{(0)} \propto \sqrt{w_j}\widehat{I}_j^{(0)}$ and $\beta_j^{(0)} = \widehat{I}_j^{(0)}$.

\subsubsection{The while loop and the stopping criterion}
Next, the while loop consists in adaptively updating the parameters $\left(\boldsymbol{\lambda}_j\right)_{j\in[\![1,J]\!]}$, $\boldsymbol{\alpha}$ and $\left(\beta_j\right)_{j\in[\![1,J]\!]}$. To do so, in the same way as in the adaptive multiple IS algorithm presented in \cite{cornuet2012adaptive}, we use all the previous samples generated so far. Before the beginning of iteration $k\geq1$, we have already generated $k$ samples $\left(\mathbf{X}^{(0,n)}\right)_{n\in[\![1,N_0]\!]},\left(\mathbf{X}^{(1,n)}\right)_{n\in[\![1,N_1]\!]},\dots,\left(\mathbf{X}^{(k-1,n)}\right)_{n\in[\![1,N_{k-1}]\!]}$, respectively according to $h_0,g_{\boldsymbol{\alpha}^{(1)}},\dots,g_{\boldsymbol{\alpha}^{(k-1)}}$. We can then consider heuristically that the concatenated sample has been generated according to the mixture $h_{k-1} \propto N_0h_0 + \sum_{i=1}^{k-1}N_ig_{\boldsymbol{\alpha}^{(i)}}$, which will be useful for the following estimations.

We first compute the new parameters $\left(\boldsymbol{\lambda}_j^{(k)}\right)_{j\in[\![1,J]\!]}$ of the IS auxiliary distribution approaching the optimal distributions $\left(g_j^*\right)_{j\in[\![1,J]\!]}$. We do so by solving the cross-entropy problem in Equation \eqref{eq:cross-entropy_max}. As explained in Section \ref{ssec:importance_sampling}, we will solve it using the stochastic counterpart with the available sample distributed according to $h_{k-1}$. Thus, in order to estimate the expectation in Equation \eqref{eq:cross-entropy_max}, it is necessary to rewrite it as an expectation over $h_{k-1}$: \begin{equation}
    \forall j \in [\![1,J]\!], \ \boldsymbol{\lambda}_j^{(k)} = \underset{\boldsymbol{\lambda}\in\Lambda}{\mathrm{argmax}} \ \mathbb{E}_{h_{k-1}}\left[\phi_j\left(\mathbf{X}\right)\log\left(g_{\boldsymbol{\lambda}}\left(\mathbf{X}\right)\right)\dfrac{f_j\left(\mathbf{X}\right)}{h_{k-1}\left(\mathbf{X}\right)}\right].
\end{equation} The corresponding stochastic counterpart problem to solve is then given by: \begin{equation}
\label{eq:update_lambda_k}
    \boldsymbol{\lambda}_j^{(k)} = \underset{\boldsymbol{\lambda}\in\Lambda}{\mathrm{argmax}} \ \sum_{i=0}^{k-1} \sum_{n=1}^{N_i} \phi_j\left(\mathbf{X}^{(i,n)}\right)\log\left(g_{\boldsymbol{\lambda}}\left(\mathbf{X}^{(i,n)}\right)\right)\dfrac{f_j\left(\mathbf{X}^{(i,n)}\right)}{h_{k-1}\left(\mathbf{X}^{(i,n)}\right)}.
\end{equation} 

We second compute the new vector $\boldsymbol{\alpha}^{(k)}\in S_J$. As explained in Section \ref{ssec:theoretical_motivations}, we will do so by solving the convex optimisation problem in Equation \eqref{eq:optimiation_alpha_k}, with the current values of the control parameters $\left(\widehat{\beta}_j^{(k-1)}\right)_{j\in[\![1,J]\!]}$. Practically, we have to estimate each variance in the sum, again with the available sample distributed according to $h_{k-1}$. The computation developed in Appendix \ref{app:optimisation_convex} shows that solving the problem in Equation \eqref{eq:optimiation_alpha_k} is equivalent to solve the following convex optimisation problem: \begin{equation}\label{eq:optimisation_alpha_k_expectation}
    \boldsymbol{\alpha}^{(k)} = \underset{\boldsymbol{\alpha}\in S_J}{\mathrm{argmin}} \ \mathbb{E}_{h_{k-1}}\left[\dfrac{\sum_{j=1}^Jw_j\left(\phi_j\left(\mathbf{X}\right)f_j\left(\mathbf{X}\right) - \widehat{\beta}_j^{(k-1)}g_{\boldsymbol{\lambda}_j^{(k)}}\left(\mathbf{X}\right)\right)^2}{g_{\boldsymbol{\alpha}}\left(\mathbf{X}\right)h_{k-1}\left(\mathbf{X}\right)}\right].
\end{equation} The corresponding stochastic counterpart problem to solve is then given by: \begin{equation}
\label{eq:update_alpha_k}
\boldsymbol{\alpha}^{(k)} = \underset{\boldsymbol{\alpha}\in S_J}{\mathrm{argmin}} \ \sum_{i=0}^{k-1} \sum_{n=1}^{N_i} \dfrac{\sum_{j=1}^J w_j\left(\phi_j\left(\mathbf{X}^{(i,n)}\right)f_j\left(\mathbf{X}^{(i,n)}\right) - \widehat{\beta}_j^{(k-1)}g_{\boldsymbol{\lambda}_j^{(k)}}\left(\mathbf{X}^{(i,n)}\right)\right)^2}{g_{\boldsymbol{\alpha}}\left(\mathbf{X}^{(i,n)}\right)h_{k-1}\left(\mathbf{X}^{(i,n)}\right)}.
\end{equation} Independently of the optimisation algorithm chosen to solve this problem, we propose to use as starting point at iteration $k$ the optimum found at iteration $k-1$, which is $\boldsymbol{\alpha}^{(k-1)}$. We compute then the new mixture $g_{\boldsymbol{\alpha}^{(k)}} = \sum_{j=1}^J\alpha_j^{(k)}g_{\boldsymbol{\lambda}_j^{(k)}}$, we draw a new sample $\left(\mathbf{X}^{(k,n)}\right)_{n\in[\![1,N_k]\!]}$ according to $g_{\boldsymbol{\alpha}^{(k)}}$ and we compute the new simulated sampling mixture $h_k \propto  N_0h_0 + \sum_{i=1}^{k}N_ig_{\boldsymbol{\alpha}^{(i)}}$.

Third, we compute the new values of the control parameters $\left(\widehat{\beta}_j^{(k)}\right)_{j\in[\![1,J]\!]}\in \mathbb{R}^J$. We estimate each of them for $j\in [\![1,J]\!]$ with the following estimator of the optimal value of the control parameter in Equation \eqref{eq:beta_star_cv}: \begin{multline}
\label{eq:update_beta_k}
\widehat{\beta}_j^{(k)} = \left(\dfrac{1}{N_k-1}\sum_{n=1}^{N_k}\left(\dfrac{g_{\boldsymbol{\lambda}_j^{(k)}}\left(\mathbf{X}^{(k,n)}\right)}{g_{\boldsymbol{\alpha}^{(k)}}\left(\mathbf{X}^{(k,n)}\right)} - m_j^{(1,k)} \right)^2\right)^{-1} \\ \left(\dfrac{1}{N_k-1}\sum_{n=1}^{N_k-1}\left(\dfrac{g_{\boldsymbol{\lambda}_j^{(k)}}\left(\mathbf{X}^{(k,n)}\right)}{g_{\boldsymbol{\alpha}^{(k)}}\left(\mathbf{X}^{(k,n)}\right)} - m_j^{(1,k)}\right)\left(\dfrac{\phi_j\left(\mathbf{X}^{(k,n)}\right)f_j\left(\mathbf{X}^{(k,n)}\right)}{g_{\boldsymbol{\alpha}^{(k)}}\left(\mathbf{X}^{(k,n)}\right)} - m_j^{(2,k)} \right)\right),
\end{multline} where \begin{equation}
    m_j^{(1,k)} = \dfrac{1}{N_k}\sum_{n=1}^{N_k}\dfrac{g_{\boldsymbol{\lambda}_j^{(k)}}\left(\mathbf{X}^{(k,n)}\right)}{g_{\boldsymbol{\alpha}^{(k)}}\left(\mathbf{X}^{(k,n)}\right)} \mbox{ and } m_j^{(2,k)} = \dfrac{1}{N_k}\sum_{n=1}^{N_k}\dfrac{\phi_j\left(\mathbf{X}^{(k,n)}\right)f_j\left(\mathbf{X}^{(k,n)}\right)}{g_{\boldsymbol{\alpha}^{(k)}}\left(\mathbf{X}^{(k,n)}\right)}.
\end{equation} Note that we choose here to use only the last sample drawn according to $g_{\boldsymbol{\alpha}^{(k)}}$ in order to make the estimation process easier, because the covariance and the variance operators in Equation \eqref{eq:beta_star_cv} are computed according to $g_{\boldsymbol{\alpha}^{(k)}}$.

Finally, we decide to stop the while loop when the final value of the criterion to minimise in Equation \eqref{eq:criterion} does not decrease anymore between two successive iterations, and more precisely when the following inequality is satisfied: \begin{multline}
\label{eq:stopping_criterion}
\dfrac{1}{N_{max}-N_0-\dots- N_{k-1}}\sum_{j=1}^Jw_j\mathbb{V}_{g_{\boldsymbol{\alpha}^{(k-1)}}}\left[\dfrac{\phi_j\left(\mathbf{X}\right)f_j\left(\mathbf{X}\right) - \widehat{\beta}_j^{(k-1)}g_{\boldsymbol{\lambda}_j^{(k-1)}}\left(\mathbf{X}\right)}{g_{\boldsymbol{\alpha}^{(k-1)}}\left(\mathbf{X}\right)}\right] \leq \\ \dfrac{1}{N_{max}-N_0-\dots- N_{k}}\sum_{j=1}^Jw_j\mathbb{V}_{g_{\boldsymbol{\alpha}^{(k)}}}\left[\dfrac{\phi_j\left(\mathbf{X}\right)f_j\left(\mathbf{X}\right) - \widehat{\beta}_j^{(k)}g_{\boldsymbol{\lambda}_j^{(k)}}\left(\mathbf{X}\right)}{g_{\boldsymbol{\alpha}^{(k)}}\left(\mathbf{X}\right)}\right].
\end{multline} This inequality compares at the end of iteration $k$ the final value of the criterion in \eqref{eq:criterion} that we would get if we had stopped the while loop after iteration $k-1$ with its value after iteration $k$. In the inequality, $N_{max}-N_0-\dots- N_{k-1}$ is the size of the independent sample used to estimate the integrals if the while loop is stopped at step $k-1$, and $N_{max}-N_0-\dots- N_{k}$ is similar for a stop at step $k$. If the inequality in Equation \eqref{eq:stopping_criterion} is satisfied, we consider that having paid a budget $N_k$ to refine the parameters from step $k-1$ to $k$ was not worth it: it would have been better to allocate this budget $N_k$ to the final estimates of the integrals, using the parameters of step $k-1$. In practice, the empirical counterpart of Equation \eqref{eq:stopping_criterion} is evaluated with the samples $\left(\mathbf{X}^{(k-1,n)}\right)_{n\in[\![1,N_{k-1}]\!]}$ and $\left(\mathbf{X}^{(k,n)}\right)_{n\in[\![1,N_{k}]\!]}$ for the left and right-hand side respectively.



\subsubsection{Final estimate with a new independent sample}

At last, at the end of the while loop after $k$ iterations, there are $N_f = N_{max} - N_0 - \dots - N_k$ calls to the functions remaining. We draw then a final i.i.d sample $\left(\mathbf{X}^{(n)}\right)_{n\in[\![1,N_f]\!]}$ according to the final sampling distribution $g_{\boldsymbol{\alpha}^{(k)}}$ which is independent, conditionally to $\boldsymbol{\alpha}^{(k)}$, $\left(\boldsymbol{\lambda}_j^{(k)}\right)_{j\in[\![1,J]\!]}$ and $\left(\widehat{\beta}_j^{(k)}\right)_{j\in[\![1,J]\!]}$, from all the previous ones drawn so far in order to get unbiased estimates $\left(\widehat{I}_{g_{\boldsymbol{\alpha}^{(k)}},g_{\boldsymbol{\lambda}_j^{(k)}},\widehat{\beta}_j^{(k)},N_f}^{\text{CV}}\right)_{j\in[\![1,J]\!]}$ of the expectations $\left(I_j\right)_{j\in[\![1,J]\!]}$, as remarked in Section \ref{ssec:control_variates}. Algorithm \ref{algo:algorithm} illustrates how to implement the described ME-aISCV algorithm in practice.

\begin{algorithm}
\caption{ME-aISCV algorithm for estimating $J$ expectations with the same $N$-sample}
\begin{algorithmic}[1]
\REQUIRE $\left(\phi_j\right)_{j\in[\![1,J]\!]}, \left(f_j\right)_{j\in[\![1,J]\!]}, \left(w_j\right)_{j\in[\![1,J]\!]},N_{max},\left(N_k\right)_{k\in\mathbb{N}}$
\STATE set $h_0 = J^{-1}\sum_{j=1}^Jf_j$ or $h_0 = \left(\sum_{j=1}^Jw_j\right)^{-1}\sum_{j=1}^Jw_jf_j$ and draw $\left(\mathbf{X}^{(0,n)}\right)_{n\in[\![1,N_0]\!]}$ according to $h_0$
\STATE for $j\in[\![1,J]\!]$, estimate $$\widehat{I}_j^{(0)} = \dfrac{1}{N_0} \sum_{n=1}^{N_0}\phi_j\left(\mathbf{X}^{(0,n)}\right)\dfrac{f_j\left(\mathbf{X}^{(0,n)}\right)}{h_0\left(\mathbf{X}^{(0,n)}\right)}$$
\STATE set $\alpha_j^{(0)} \propto \sqrt{w_j}\widehat{I}_j^{(0)}$ and $\widehat{\beta}_j^{(0)} = \widehat{I}_j^{(0)}$
\STATE set $N_{eval}=N_0$ and $k=0$
\WHILE{$N_{eval}<N_{max}/2$}
\STATE update $k = k+1$
\STATE for $j\in[\![1,J]\!]$, estimate the new distribution parameters $\boldsymbol{\lambda}_j^{(k)}$ by solving the cross-entropy problem in Equation \eqref{eq:update_lambda_k}
\STATE estimate $\boldsymbol{\alpha}^{(k)}$ by solving the optimisation problem in Equation \eqref{eq:update_alpha_k} using as starting point $\boldsymbol{\alpha}^{(k-1)}$
\STATE set $g_{\boldsymbol{\alpha}^{(k)}} = \sum_{j=1}^J\alpha_j^{(k)}g_{\boldsymbol{\lambda}_j^{(k)}}$
\STATE draw $\left(\mathbf{X}^{(k,n)}\right)_{n\in[\![1,N_k]\!]}$ according to $g_{\boldsymbol{\alpha}^{(k)}}$ and update $N_{eval} = N_{eval}+N_k$
\STATE update $h_k = \frac{N_{eval} - N_k}{N_{eval}}h_{k-1} + \frac{N_k}{N_{eval}}g_{\boldsymbol{\alpha}^{(k)}}$
\STATE for $j\in[\![1,J]\!]$, estimate $\widehat{\beta}_j^{(k)}$ with Equation \eqref{eq:update_beta_k}
\IF{the stopping criterion in Equation \eqref{eq:stopping_criterion} is satisfied}
\STATE exit the while loop 
\ENDIF

\ENDWHILE

\STATE set $N_f = N_{max} - N_{eval}$
\STATE draw $\left(\mathbf{X}^{(n)}\right)_{n\in[\![1,N_f]\!]}$ according to $g_{\boldsymbol{\alpha}^{(k)}}$

\RETURN $$\widehat{I}_{g_{\boldsymbol{\alpha}^{(k)}},g_{\boldsymbol{\lambda}_j^{(k)}},\widehat{\beta}_j^{(k)},N_f}^{\text{CV}} = \dfrac{1}{N_f}\sum_{n=1}^{N_f}\dfrac{\phi_j\left(\mathbf{X}^{(n)}\right)f_j\left(\mathbf{X}^{(n)}\right) - \widehat{\beta}_j^{(k)}g_{\boldsymbol{\lambda}_j^{(k)}}\left(\mathbf{X}^{(n)}\right)}{g_{\boldsymbol{\alpha}^{(k)}}\left(\mathbf{X}^{(n)}\right)} + \widehat{\beta}_j^{(k)}$$
\end{algorithmic}
\label{algo:algorithm}
\end{algorithm}

\FloatBarrier

\section{Applications to sensitivity analysis and numerical results}
\label{sec:numerical_results}

In order to illustrate the practical interest of the previous efforts, this section aims to evaluate numerically the performances of the suggested ME-aISCV algorithm to estimate $J$ expectations with the same sample, and to compare them to the performances of the existing methods. The code to reproduce the numerical experiments is publicly available at: \url{https://github.com/Julien6431/Multiple\_expectation\_estimation.git}.

Let us introduce the adopted numerical parameters that will be used: \begin{itemize}
    \item $N_{max} = 2\times 10^4$ which represents the total number of calls to the functions,
    \item for all $k\in\mathbb{N}$, we choose $N_k = N_{max}/10 = 2\times 10^3$,
    \item each of the IS auxiliary distribution $g_{\boldsymbol{\lambda}}$ will be picked in the Gaussian family,
    \item we use the Sequential Least SQuares Programming (SLSQP) algorithm \cite{kraft1988software} to solve the convex problem in Equation \eqref{eq:optimiation_alpha_k}, because it is well-suited for bounded and constrained problems,
    \item $n_{rep} = 200$ realisations of each estimator to represent the results as boxplots.
\end{itemize}

For adaptive algorithms, a discussion about the choice of the sequence $\left(N_k\right)_{k\in\mathbb{N}}$ is made in \cite{cornuet2012adaptive}. At first, it can be more intuitive to consider a sequence that increases with the accuracy of the IS auxiliary distributions. However, it is difficult to recover from poor early samples because of the "what-you-get-is-what-you-see" nature of these kind of algorithms. Therefore, as said in \cite{cornuet2012adaptive}, a good trade-off is then to consider a stationary sequence, as we do here.

\subsection{Estimation of the non-centered moments of the standard Gaussian distribution}
\label{ssec:moment_estimation}

First, for illustration purposes, let us consider the simple problem of the estimation of the non-centered even moments of the one-dimensional standard Gaussian distribution. More precisely, the expectations to estimate are defined by $\left(I_j^{\text{mom}} = \mathbb{E}_{f_1}\left(X^{2j}\right)\right)_{j\in[\![1,J]\!]}$, where $f_1$ is the PDF of the standard Gaussian distribution $\mathcal{N}_1\left(0,1\right)$. Note that we consider only the even moments between $2$ and $2J$ for two reasons: first, since the standard Gaussian distribution is symmetric around zero, its odd moments are equal to $0$, and second, the functions of interest must be non negative, as defined in Section \ref{ssec:general_presentation}.

We consider here $J=10$, and reference values are computed with their analytical expressions. We compare the performances of the proposed algorithm with the ones of the classical Monte Carlo estimations. For pedagogical purposes, as the theoretical values are known, we set $w_j = \left(I_j^{\text{mom}}\right)^{-2}$ for all $j\in[\![1,J]\!]$ in Equation \eqref{eq:criterion}. Numerical results are presented graphically in Figure \ref{fig:gaussian_moments}. The boxplots show that the quality of the estimations of the $J=10$ expectations is significantly better with the ME-aISCV algorithm than with the existing Monte Carlo method. These observations are confirmed by Table \ref{table:gaussian_moments}, because the criterion to minimize has been divided by about $10^4$. Note that for the moments of order $16$, $18$ and $20$, the Gaussian approximation of the standard Monte Carlo estimation does not kick-in at all. As a result, although the estimation is unbiased, its distribution is highly asymmetric and its median is far from its mean.

\begin{figure}
    \centering
    \includegraphics[width=\textwidth]{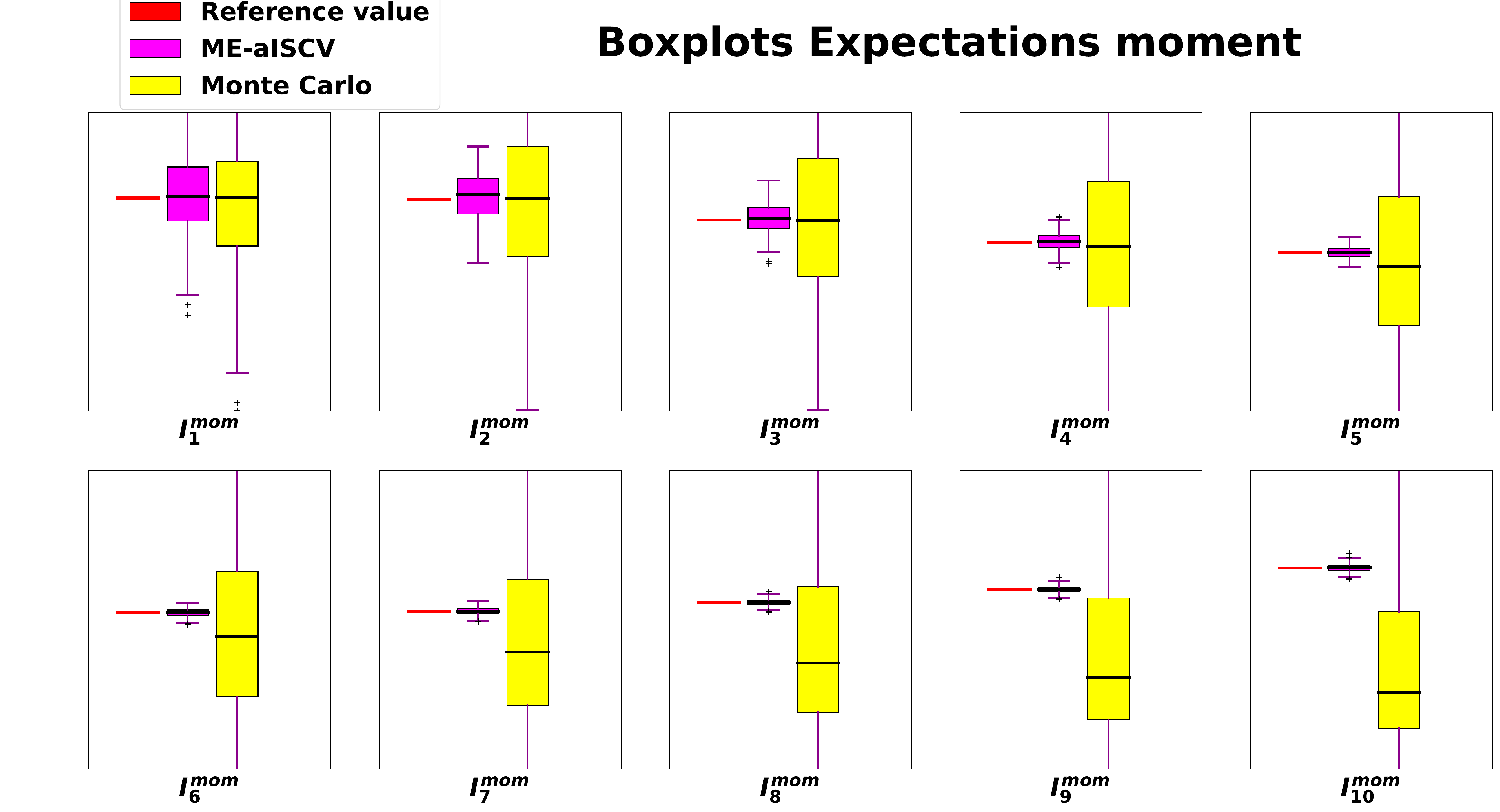}
    \caption{Estimation of the $J=10$ first even moments of the one-dimensional standard Gaussian distribution.}
    \label{fig:gaussian_moments}
\end{figure}

\begin{table}[h!]
\centering
\begin{tabular}{|c|c|c|} 
 \hline
& Monte-Carlo & ME-aISCV\\
 \hline
 $\sum_{j=1}^Jw_j\mathbb{V}\left(\widehat{I}_j^{\text{mom}}\right)$ & $12.782$ & $1.631\times10^{-3}$\\ 
\hline
\end{tabular}
\caption{Weighted sum of the variances of the estimators of the $J=10$ first even moments of the one-dimensional standard Gaussian distribution.}
\label{table:gaussian_moments}
\end{table}

Figure \ref{fig:gaussian_moments_pdf_evolution} represents the evolution of the distribution $g_{\boldsymbol{\alpha}_k}$ during the procedure for one execution of the ME-aISCV algorithm. The optimal IS distribution is a mixture of the distribution $g_{2j}^*\left(x\right) \propto x^{2j}f_1\left(x\right)$ for $j\in[\![1,J]\!]$. In particular, it is symmetric around zero and its standard deviation might be larger than $1$. First, the blue line represents the PDF of the initial distribution. Then, the orange line represents the PDF of the mixture $g_{\boldsymbol{\alpha}_1}$ obtained at the end of iteration $1$. We can see in particular that it is not symmetric around zero, and so it is not close to the target sampling distribution. Next, the green line represents the PDF of the mixture $g_{\boldsymbol{\alpha}_2}$ obtained at the end of iteration $2$. It is now symmetric around zero and is then a good candidate. However, another iteration is necessary because the stopping criterion in Equation \eqref{eq:stopping_criterion} is not reached yet. At last, the red line represents the PDF of the mixture $g_{\boldsymbol{\alpha}_3}$ obtained at the end of iteration $3$. It is very close to the green line, so the third iteration did not improve a lot the accuracy of the IS sampling distribution and the stopping criterion is thus reached. The distribution $g_{\boldsymbol{\alpha}_3}$ is then the final IS sampling distribution and the while loop is over in Algorithm \ref{algo:algorithm}.

\begin{figure}
    \centering
    \includegraphics[width=\textwidth]{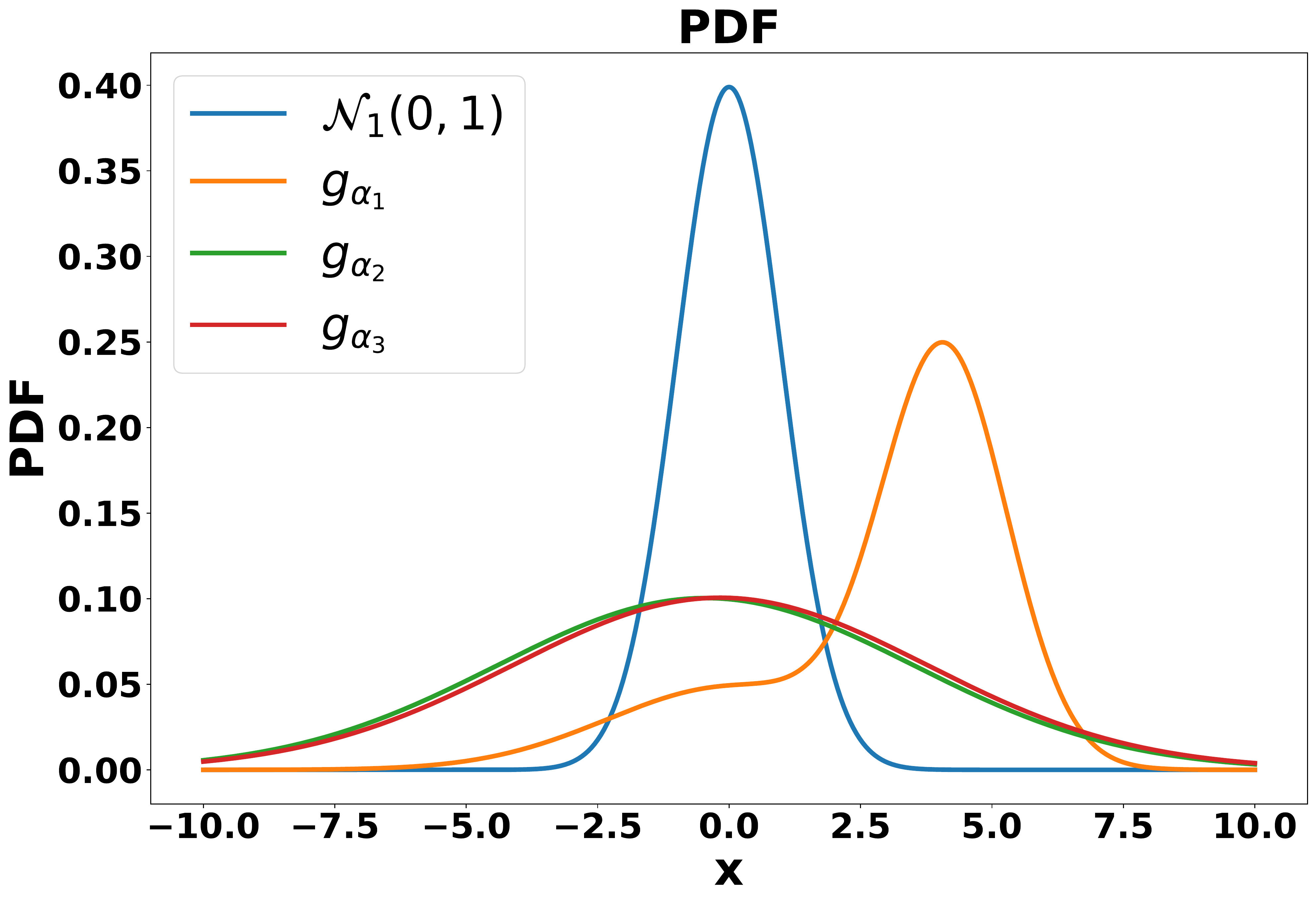}
    \caption{Evolution of the distribution $g_{\boldsymbol{\alpha}_k}$ during one execution of the algorithm.}
    \label{fig:gaussian_moments_pdf_evolution}
\end{figure}

\subsection{Estimation of Sobol' indices}
\label{ssec:sobol_estimation}

\subsubsection{Presentation of the problem}
\label{ssec:presentation_sobol}

The Sobol' indices \cite{sobol1993sensitivity} are quantitative tools which allow to quantify the influence of each input variable on the variability of the output, in the case where the input variables are mutually independent. For all $i\in[\![1,d]\!]$, the first order Sobol' indices are defined, for a function $\phi:\mathbb{X}\longrightarrow\mathbb{R}_+$, by: \begin{equation}\label{eq:sobol_indices} S_i = \dfrac{\mathbb{V}_{f}\left[\mathbb{E}_{f}\left(\phi\left(\mathbf{X}\right)|X_i\right)\right]}{\mathbb{V}_{f}\left(\phi\left(\mathbf{X}\right)\right)}.
\end{equation} We will estimate them with the well-known Pick-Freeze method introduced in \cite{sobol1993sensitivity,homma1996importance}. It consists in rewriting each Sobol' index in Equation \eqref{eq:sobol_indices} as a single expectation. The idea is to introduce a second random variable $\mathbf{X}^{i} = \left(X_i,\mathbf{X}_{-i}'\right)$, where $\mathbf{X}_{-i}' = \left(X_1',\dots,X_{i-1}',X_{i+1}',\dots,X_d'\right)$ satisfies $\mathbf{X}_{-i}' \overset{d}{=}\mathbf{X}_{-i}$ and $\mathbf{X}_{-i}'\perp\!\!\!\perp\mathbf{X}_{-i}$ and where $\perp\!\!\!\perp$ is the independence symbol. By decomposing the variance at the denominator as well, the Sobol' indices can be then rewritten for all $i\in[\![1,d]\!]$ as: \begin{equation}\label{eq:pick_freeze}S_i = \dfrac{\mathbb{E}_{f}\left(\phi(\mathbf{X})\phi(\mathbf{X}^i)\right) - \mathbb{E}_{f}\left(\phi(\mathbf{X})\right)^2}{\mathbb{E}_{f}\left(\phi(\mathbf{X})^2\right) - \mathbb{E}_{f}\left(\phi(\mathbf{X})\right)^2}.\end{equation} This procedure requires then $N_{\text{PF}} = N(d+1)$ calls to the function $\phi$ to compute the $d$ first order Sobol' indices.

\subsubsection{Formulation as a multiple estimation problem}

To estimate the $d$ first order Sobol' indices in Equation \eqref{eq:pick_freeze}, there are $J = d+2$ different expectations to estimate: the Pick-Freeze expectations $\mathbb{E}_{f}\left(\phi(\mathbf{X})\phi(\mathbf{X}^i)\right)$ for $i\in [\![1,d]\!]$, $\mathbb{E}_{f}\left(\phi(\mathbf{X})\right)$ and $\mathbb{E}_{f}\left(\phi(\mathbf{X})^2\right)$. The classical method to estimate them by Pick-Freeze consists in drawing two independent i.i.d $N$-samples according to $f$ and to mix both of them to build the random variables $\mathbf{X}$ and $\mathbf{X}^i$ for $i\in [\![1,d]\!]$. This process is equivalent to considering the augmented space $\mathbb{X}\times\mathbb{X}$ of dimension $2d$, to draw an i.i.d. $N$-sample according to the distribution of PDF $\widetilde{f} : \left(\mathbf{x},\mathbf{x}'\right) \in \mathbb{X}\times\mathbb{X} \mapsto f\left(\mathbf{x}\right)\times f\left(\mathbf{x}'\right)$ and to make the appropriate combinations to build the random variables $\mathbf{X}$ and $\mathbf{X}^i$ for $i\in [\![1,d]\!]$. The corresponding functions in the augmented space are then:\begin{equation}
    \begin{array}{l|rcl}
\phi_i : & \mathbb{X}\times\mathbb{X} & \longrightarrow & \mathbb{R} \\
    & \left(\mathbf{x},\mathbf{x}'\right) & \longmapsto & \phi\left(x_i,\mathbf{x}_{-i}\right)\phi\left(x_i,\mathbf{x}' _{-i}\right),
    \end{array} 
\end{equation} and \begin{equation}
    \begin{array}{l|rcl}
\phi_{d+1} : & \mathbb{X}\times\mathbb{X} & \longrightarrow & \mathbb{R} \\
    & \left(\mathbf{x},\mathbf{x}'\right) & \longmapsto & \phi\left(\mathbf{x}\right)
    \end{array} \mbox{ and } \begin{array}{l|rcl}
\phi_{d+2} : & \mathbb{X}\times\mathbb{X} & \longrightarrow & \mathbb{R} \\
    & \left(\mathbf{x},\mathbf{x}'\right) & \longmapsto & \phi\left(\mathbf{x}\right)^2.
    \end{array}
\end{equation}

Finally, we have here a family $\left(\mathbb{E}_{\widetilde{f}}\left(\phi_i\left(\mathbf{X},\mathbf{X}'\right)\right)\right)_{i\in[\![1,d+2]\!]}$ of $J=d+2$ different expectations to estimate under the same input distribution $\widetilde{f}$, which corresponds to the Case 1 presented in Section \ref{ssec:presentation_problem}. All the weights $\left(w_j\right)_{j\in[\![1,J]\!]}$ are set to $1$.

\subsubsection{Numerical results on the cantilever beam problem}
\label{sssec:numerical_resuls_sobol}

The cantilever beam problem is a real structure engineering problem which is presented in \cite{zhou2014moment,baoyu2017reliability}. Consider a rectangular cantilever beam structure. The dimensional parameters of the beam are denoted $l_X$, $l_Y$ and $L$. The elastic modulus of the structure is represented by $E$. Two random forces $F_X$ and $F_Y$ are exerted on the tip of the section. The goal function is the maximum vertical displacement of the tip section, which is given analytically according to the previous parameters by: \begin{equation}\label{eq:maximal_displacement}
    \phi\left(F_X,F_Y,E,l_X,l_Y,L\right) = \dfrac{4L^3}{10^9\times El_Xl_Y}\sqrt{\left(\dfrac{F_X}{l_X^2}\right)^2+\left(\dfrac{F_Y}{l_Y^2}\right)^2}.
\end{equation} The distributions of each input variable are listed in Table \ref{table:distr_cantiler_beam}. \begin{table}[h!]
\centering
\begin{tabular}{|c | l l l l|} 
 \hline
 & Symbol and Unit & Distribution & Mean & Coefficient of variation\\ [0.5ex] 
 \hline\hline
 1 & $F_X$ (N) & LogNormal & $m_1$ & $0.08$\\ 
 2 & $F_Y$ (N) & LogNormal & $m_2 $ & $0.08$\\
 3 & $E$ (Pa) & LogNormal & $m_3 $ & $0.06$\\
 4 & $l_X$ (m)& Normal & $m_4$ & $0.1$\\ 
 5 & $l_Y$ (m)& Normal & $m_5 $ & $0.1$\\ 
 6 & $L$ (m)& Normal & $m_6$ & $0.1$ \\[1ex] 
 \hline
\end{tabular}
\caption{Distributions of each input variable of the cantilever beam example}
\label{table:distr_cantiler_beam}
\end{table} Moreover, the dimensional variables $l_X$, $l_Y$ and $L$ are linearly dependent through the following Pearson correlation coefficients:
\begin{equation}
    \rho_{l_X,l_Y} = m_7 \mbox{ and } \rho_{L,l_X} = m_8  \mbox{ and } \rho_{L,l_Y} = m_9.
\end{equation} This input distribution is parameterized by the sequence of parameters $\boldsymbol{m} = \left(m_i\right)_{i\in[\![1,9]\!]}\in\mathbb{R}_+^3 \times \mathbb{R}^3 \times ]-1,1[^3$.

We want to estimate the first order Sobol' indices in Equation \eqref{eq:sobol_indices} for this system. Here, the input distribution is fully known and the parameter $\boldsymbol{m}$ is given by $\boldsymbol{m}_{sob} = \left(556.8,453.6,200,0.062,0.0987,4.29,0,0,0\right)$. In line with Section \ref{ssec:presentation_sobol}, all the input variables are independent because the three Pearson correlation coefficients $\rho_{l_X,l_Y}$, $\rho_{L,l_X}$ and $\rho_{L,l_Y}$ are assumed to be equal to 0 in this section, which is a necessary assumption for the Sobol' indices to have their full set of beneficial properties.

References values of the Sobol' indices are obtained by applying the existing Pick-Freeze estimation scheme with two $N$-samples of (very large) size $N=10^7$. Moreover, we compare the performances of the ME-aISCV algorithm with the ones of the existing standard Pick-Freeze estimation scheme using two $N_{max}$-samples such that both methods require exactly the same number $N_{\text{PF}}$ of calls to the function $\phi$. 

The results of the estimations of the first order Sobol' indices for the cantilever beam problem are given in Figure \ref{fig:cantilever_beam_sobol}. We can see that the ME-aISCV algorithm provides significantly better performances than the existing method for estimating the Sobol' indices. Indeed, the boxplots corresponding to the ME-aISCV algorithm are centered on the reference values and have a much smaller stretch. These observations are confirmed by the numerical values in Tables \ref{table:variance_individual_sobol} and \ref{table:cantilever_beam_criterion}. The individual variances of each estimator of the first order Sobol' indices are divided by $10$ and consequently the sum of the variances.

\begin{figure}
    \centering
    \includegraphics[width=\textwidth]{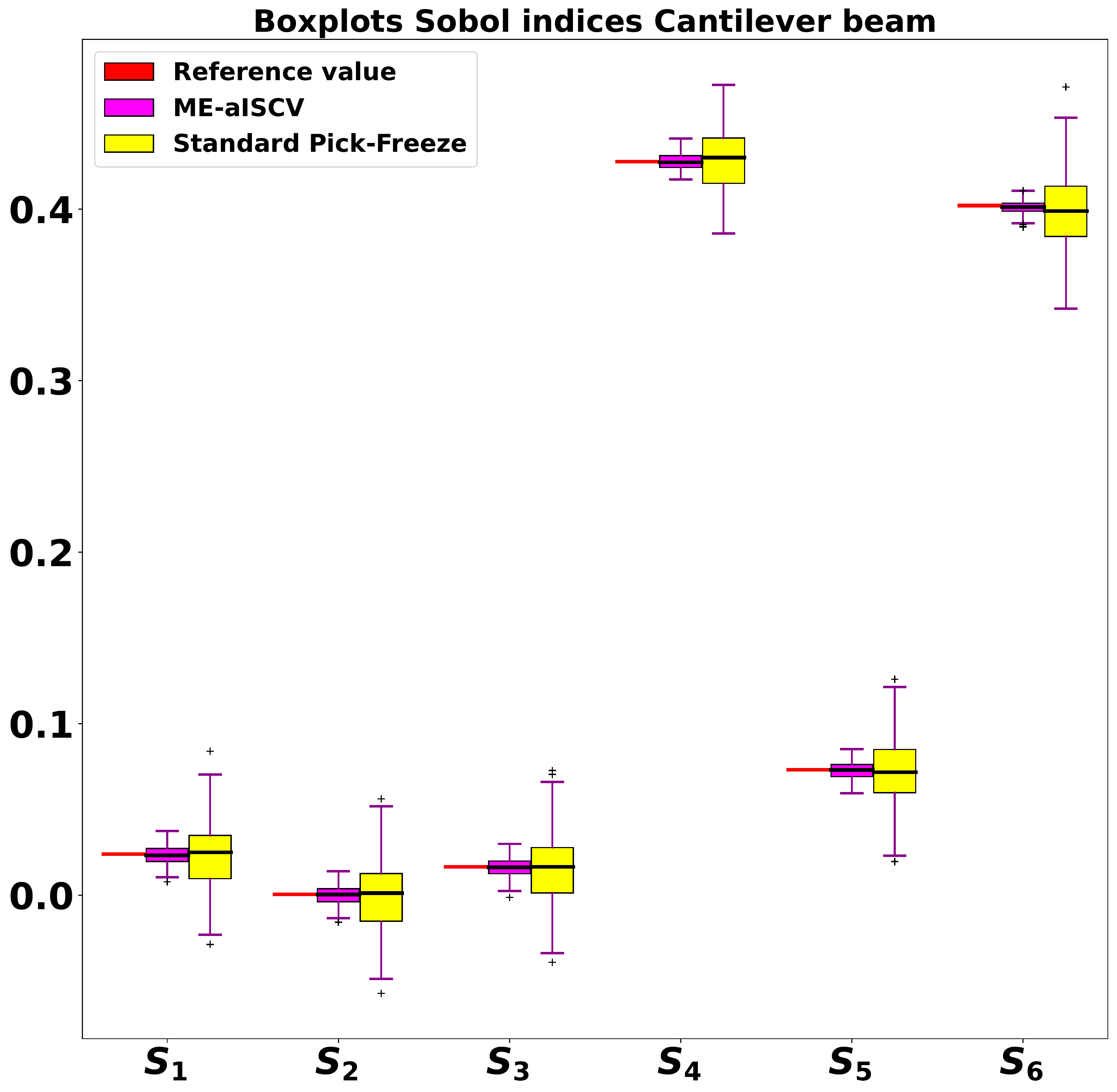}
    \caption{Estimation of the Sobol' indices for the cantilever beam problem with independent input variables.}
    \label{fig:cantilever_beam_sobol}
\end{figure}

\begin{table}[h!]
\centering
\begin{tabular}{|c|c|c|} 
 \hline
& standard Pick-Freeze & ME-aISCV\\
 \hline
 $\mathbb{V}\left(\widehat{S}_1\right)$ & $4.316\times10^{-4}$ & $3.372\times10^{-5}$\\ 
\hline
$\mathbb{V}\left(\widehat{S}_2\right)$ & $4.375\times10^{-4}$ & $3.364\times10^{-5}$\\ 
\hline
$\mathbb{V}\left(\widehat{S}_3\right)$ & $4.412\times10^{-4}$ & $3.308\times10^{-5}$\\ 
\hline
$\mathbb{V}\left(\widehat{S}_4\right)$ & $3.635\times10^{-4}$ & $2.377\times10^{-5}$\\ 
\hline
$\mathbb{V}\left(\widehat{S}_5\right)$ & $4.112\times10^{-4}$ & $3.053\times10^{-5}$\\ 
\hline
$\mathbb{V}\left(\widehat{S}_6\right)$ & $4.605\times10^{-4}$ & $1.943\times10^{-5}$\\ 
\hline
\end{tabular}
\caption{Individual variance of each of the $d=6$ estimators of the first order Sobol' indices for both methods.}
\label{table:variance_individual_sobol}
\end{table}

\begin{table}[h!]
\centering
\begin{tabular}{|c|c|c|} 
 \hline
& Monte-Carlo & ME-aISCV\\
 \hline
 $\sum_{i=1}^d\mathbb{V}\left(\widehat{S}_i\right)$ & $2.533\times10^{-3}$ & $1.733\times10^{-4}$\\ 
\hline
\end{tabular}
\caption{Sum of the variances of the estimators of the $d=6$ first order Sobol' indices for both methods.}
\label{table:cantilever_beam_criterion}
\end{table}

\FloatBarrier

\subsection{Sensitivity analysis w.r.t. parameters of the input distribution}
\label{ssec:sensitivity_wrt_input_params}

\subsubsection{Presentation of the problem}
Most of the time, the input distribution of a computer model $\phi$ is assumed to be fully known and determined. However, this assumption is not always true in practice. Indeed, because of lack of knowledge or data, the input distribution might depend on unknown or uncertain parameters $\boldsymbol{m}$, such as the mean vector or the standard deviations of the marginals for example. This epistemic uncertainty is then also propagated through the computer model $\phi$, and can thus have an impact on the output value of the system.

In order to quantify the individual influence of the parameters in $\boldsymbol{m}$ on a quantity of interest, such as the mean of the output, a solution is to compute some sensitivity indices of the uncertain parameters, such as the Sobol' indices defined in Section \ref{ssec:sobol_estimation}.

\subsubsection{Formulation as a multiple estimation problem}

The quantity of interest considered here is the mean output value of the function. To achieve the goal presented above and estimate the sensitivity indices, one need to get an input/output dataset $\left(\boldsymbol{m}^{(j)},\mathbb{E}_{f_{\boldsymbol{m}^{(j)}}}\left(\phi\left(\mathbf{X}\right)\right)\right)_{j\in[\![1,J]\!]}$, with $\left(\boldsymbol{m}^{(j)}\right)_{j\in[\![1,J]\!]}$ a sample of $J$ sets of parameters and $\left(f_{\boldsymbol{m}^{(j)}}\right)_{j\in[\![1,J]\!]}$ its corresponding PDF family. The challenge is then to efficiently estimate each expectation $\mathbb{E}_{f_{\boldsymbol{m}^{(j)}}}\left(\phi\left(\mathbf{X}\right)\right)$ for $j\in[\![1,J]\!]$. We have then to estimate a family of $J$ expectations of the same computer model $\phi$ under $J$ different input distributions $\left(f_{\boldsymbol{m}^{(j)}}\right)_{j\in[\![1,J]\!]}$, which corresponds to the Case 2 presented in Section \ref{ssec:presentation_problem}. All the weights $\left(w_j\right)_{j\in[\![1,J]\!]}$ are set to $1$.  

\subsubsection{Numerical results on the cantilever beam problem}

Let us consider again the cantilever beam problem presented in Section \ref{sssec:numerical_resuls_sobol}. The parameter $\boldsymbol{m} = \left(m_i\right)_{i\in[\![1,9]\!]}$ is here supposed uncertain, with independent components whose marginal distributions are given in Table \ref{table:distr_m}. The quantity of interest is the mean value of the maximal vertical displacement of the tip section given in Equation \eqref{eq:maximal_displacement}.

\begin{table}[h!]
\centering
\begin{tabular}{|c | c c|} 
 \hline
 & Parameter & Distribution \\ [0.5ex] 
 \hline\hline
 1 & $m_1$ & $\mathcal{U}(525,575)$ \\ 
 2 & $m_2$ & $\mathcal{U}(425,475)$ \\ 
 3 & $m_3$ & $\mathcal{U}(175,225)$ \\ 
 4 & $m_4$ & $\mathcal{U}(0.06,0.07)$ \\ 
 5 & $m_5$ & $\mathcal{U}(0.09,0.1)$ \\ 
 6 & $m_6$ & $\mathcal{U}(4,5)$ \\ 
 7 & $m_7$ & $\mathcal{U}(-0.6,0)$ \\
 8 & $m_8$ & $\mathcal{U}(0,0.5)$ \\
 9 & $m_9$ & $\mathcal{U}(0,0.5)$ \\[1ex] 
 \hline
\end{tabular}
\caption{Marginal distributions of the random parameter $\boldsymbol{m} = \left(m_i\right)_{i\in[\![1,9]\!]}$.}
\label{table:distr_m}
\end{table}

Here, we estimate $J=100$ expectations. A sample of parameters $\left(\boldsymbol{m}^{(j)}\right)_{j\in[\![1,J]\!]}$ is drawn according to the distribution in Table \ref{table:distr_m} with the Latin Hypercube Simulation (LHS) method \cite{helton2003latin}. References values for the $J=100$ expectations are computed with the crude Monte Carlo estimator of each expectation with samples of (very large) size $N=10^7$. 
To evaluate the performances of the ME-aISCV algorithm, we compare it to two existing estimators. The first one is the naive Monte Carlo method (nMC) which consists, for $j\in[\![1,J]\!]$, in drawing an i.i.d sample of size $N_{max}/J$ according to each distribution $f_{\boldsymbol{m}^{(j)}}$ and to compute the corresponding empirical mean of the output. The second one consists in considering a unique sampling distribution $h = J^{-1}\sum_{j=1}^Jf_{\boldsymbol{m}^{(j)}}$ which is the mixture of the $J$ different input distribution and to compute the following estimators: \begin{equation}
    \widehat{I}_j^{\text{MCmixt}} = \dfrac{1}{N_{max}} \sum_{n=1}^{N_{max}}\phi\left(\mathbf{X}^{(n)}\right)\dfrac{f_{\boldsymbol{m}^{(j)}}\left(\mathbf{X}^{(n)}\right)}{h\left(\mathbf{X}^{(n)}\right)},
\end{equation} where $\left(\mathbf{X}^{(n)}\right)_{n\in[\![1,N_{max}]\!]}$ is an i.i.d. sample drawn according to $h$. The distribution $h$ corresponds then to the initial sampling distribution $h_0$ of Algorithm \ref{algo:algorithm}. Both methods require exactly $N_{max}$ calls to the function $\phi$, as the proposed algorithm.

The results of the estimations of the $J=100$ expectations for the cantilever beam problem are given in Figure \ref{fig:cantilever_beam_expectations}. We can see that the ME-aISCV algorithm provides significantly better performances than the existing methods for estimating a large number of expectations, for the same reasons as in the previous example. These observations are confirmed by the numerical values in Table \ref{table:cantilever_beam_criterion_expectations}. Indeed, the criterion to minimize has been considerably reduced with the proposed algorithm compared to the existing methods.

\begin{figure}
    \centering
    \includegraphics[width=\textwidth]{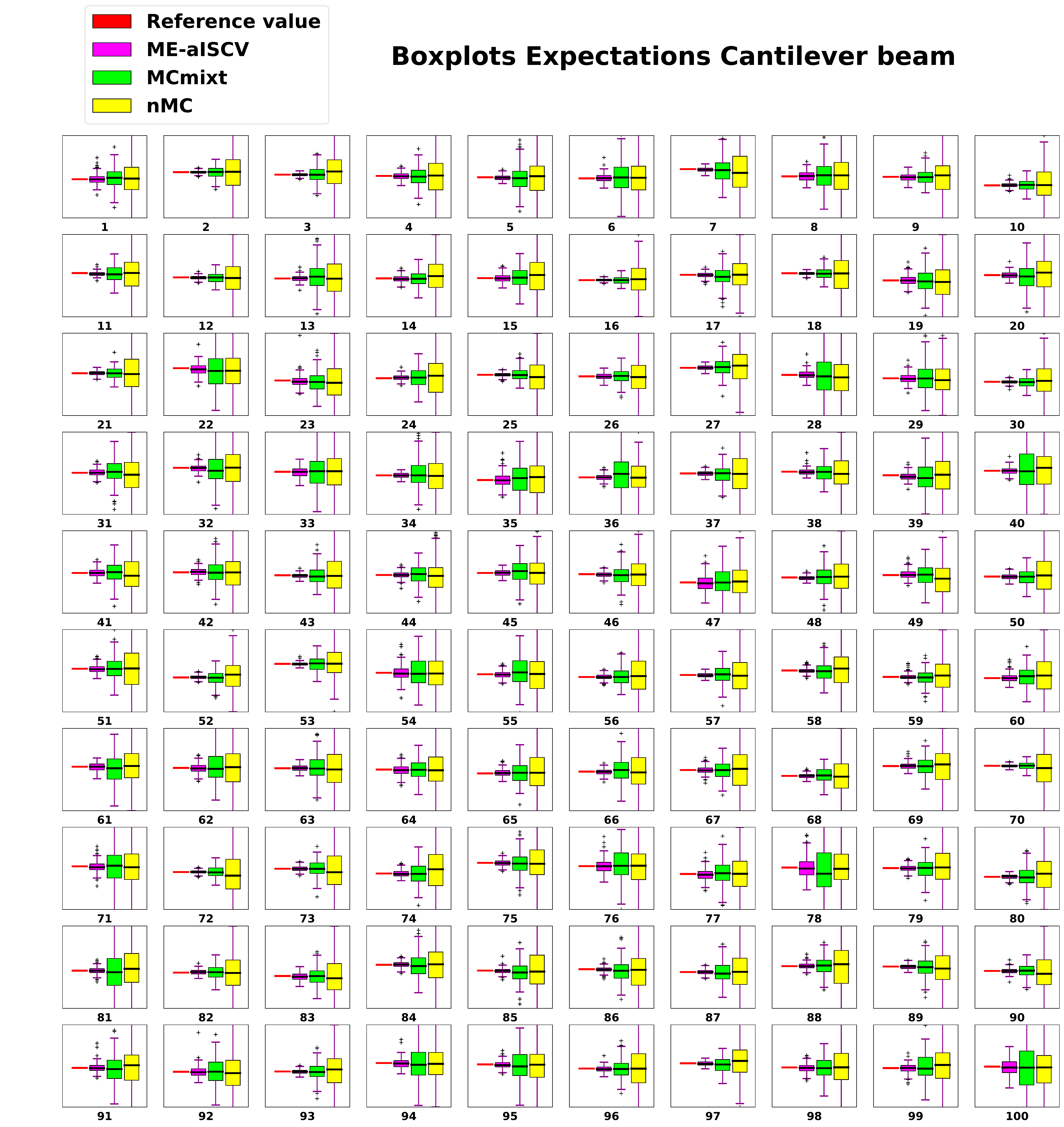}
    \caption{Estimation of the $J=100$ expectations for the cantilever beam problem.}
    \label{fig:cantilever_beam_expectations}
\end{figure}

\begin{table}[h!]
\centering
\begin{tabular}{|c|c|c|c|} 
 \hline
& nMC & MCmixt & ME-aISCV\\
 \hline
 $\sum_{j=1}^J\mathbb{V}\left(\widehat{I}_j\right)$ & $1.309\times10^{-4}$ & $6.103\times10^{-5}$ & $4.379\times10^{-6}$\\ 
\hline
\end{tabular}
\caption{Sum of the variances of the estimators of the $J=100$ expectations for all methods.}
\label{table:cantilever_beam_criterion_expectations}
\end{table}

Moreover, this example highlights a specific property of the ME-aISCV algorithm due to the choice of the criterion to minimize in Equation \eqref{eq:criterion}. One can see on Figure \ref{fig:cantilever_beam_expectations} that some expectations have more variance reduction than others, since their corresponding boxplots have a lower stretch. Indeed, due to the form of the criterion to minimize in Equation \eqref{eq:criterion}, high values of $w_j\mathbb{V}\left(\widehat{I}_j\right)$ have a more important role in the sum than lower ones. Therefore, the proposed algorithm will mainly focus on reducing before anything else the variance of the corresponding estimators, which explains the phenomenon described and observed here.

\section{Conclusion}
\label{sec:conclusion}

In the present article, we are interested in efficiently estimating multiple expectations with the same $N$-sample, a problematic encountered in some classical problems related to the study of black-box models. The criterion used to quantify the quality of the common estimation of the expectations is the weighted sum of each individual variance given in Equation \eqref{eq:criterion}. We show that there exists a family of optimal estimators combining both IS and CV, which nevertheless cannot be used in practice because they require the knowledge of the values of the expectations to estimate. Motivated by the form of these optimal estimator and some interesting properties, we suggest a new effective ME-aISCV algorithm combining both IS and CV, whose general idea is to adaptively update the IS distributions as well as the control parameters for approaching the optimal ones until a quantitative stopping criterion is reached. The main goal of this adaptive procedure is to minimize as much as possible the criterion in Equation \eqref{eq:criterion}. Then, a new independent sample drawn according to the final IS sampling distribution allows to get unbiased estimators by IS and CV of all the expectations. Finally, we illustrate and discuss the practical interest of the proposed algorithm. We first address the estimation of the even moments of the standard Gaussian distribution. Then, we show that the suggested ME-aISCV algorithm is generally applicable to sensitivity analysis, both on the input parameters and also on their uncertainty distribution. This is applied to the physical cantilever beam problem. Overall, the applications demonstrate the robustness of the algorithm to a wide range of situations. Especially, the high-order moments of the Gaussian distribution imply that the IS distributions must explore the far tails of the initial one. Furthermore, $100$ expectations are estimated simultaneously in the input-distribution-sensitivity example.

A first way of improvement of the ME-aISCV algorithm is to adaptively update the weights $\left(w_j\right)_{j\in[\![1,J]\!]}$ during the while loop in Algorithm \ref{algo:algorithm}. Indeed, it can be interesting to adjust online the importance given to each expectation or to estimate more accurately unknown target weights, such as $\left(I_j^{-2}\right)_{j\in[\![1,J]\!]}$ for example. In that latter case, the criterion in Equation \eqref{eq:criterion} is the sum of the square coefficients of variation of each estimator. Another way of improvement of this algorithm is to use non-parametric IS auxiliary distributions \cite{zhang1996nonparametric} to approach the optimal distributions $\left(g_j^*\right)_{j\in[\![1,J]\!]}$ defined at the beginning of Section \ref{sec:new_algorithm}. This method allows more flexibility and to approach more complex target distributions, but faces the curse of dimensionality. At last, the algorithm can be adapted to estimate small failure probabilities. It can be done by performing adaptive parametric IS to solve the cross-entropy problem in Equation \eqref{eq:cross-entropy_max} as in \cite{rubinstein2013cross} to approach the optimal distributions $\left(g_j^*\right)_{j\in[\![1,J]\!]}$ adapted to small failure probabilities. An interesting application of this adaptation can be found in \cite{morio2011influence} and consists in identifying the most influential parameters of the input distribution on the variability of the failure probability of the system.

Finally, a more complex application of this new method is the estimation of the Shapley effects for global sensitivity analysis with dependent input variables \cite{owen2014sobol}. Estimating each of them efficiently is a challenging task because it requires the estimation of the closed Sobol' indices for many subsets $u \subseteq [\![1,d]\!]$. A formulation of this problem as a multiple expectation estimation problem has been written in \cite{broto2020variance}, and the estimation of the Shapley effects in a reliability context by IS has been investigated in \citep{demange2022shapley}. Since the inputs are dependent, it is no longer possible to perform the estimation in the augmented space $\mathbb{X}\times\mathbb{X}$ as we did in Section \ref{ssec:sobol_estimation}. The main remaining challenge is then to find an optimal IS distribution in $\mathbb{X}$ associated to each closed Sobol' index in order to be able to apply the proposed ME-aISCV algorithm.

\section*{Acknowledgements}

The first author is enrolled in a Ph.D. program co-funded by \textit{ONERA – The French Aerospace Lab} and \textit{Toulouse III - Paul Sabatier University}. Their financial supports are gratefully acknowledged.

\section*{Appendix}

\begin{appendices}
\section{Equivalence between both optimization problem}
\label{app:optimisation_convex}
Let us prove that the optimization problem in Equation \eqref{eq:optimiation_alpha_k} is equivalent to the one in Equation \eqref{eq:optimisation_alpha_k_expectation}. Consider a sequence $\boldsymbol{\alpha}\in S_J$, a family of IS auxiliary distributions $\left(g_{\boldsymbol{\lambda}_j}\right)_{j\in[\![1,J]\!]}$, a family of control parameters $\left(\beta_j\right)_{j\in[\![1,J]\!]}\in\mathbb{R}^J$ and a family of positive weights $\left(w_j\right)_{j\in[\![1,J]\!]}\in\mathbb{R}_+^J$.

For any $j\in[\![1,J]\!]$ and any IS auxiliary distribution $h$, we have: \begin{align*}
    \mathbb{V}&_{g_{\boldsymbol{\alpha}}}\left(\dfrac{\phi_j\left(\mathbf{X}\right)f_j\left(\mathbf{X}\right)-\beta_j g_{\boldsymbol{\lambda}_j}\left(\mathbf{X}\right)}{g_{\boldsymbol{\alpha}}\left(\mathbf{X}\right)}\right) \\
    &= \mathbb{E}_{g_{\boldsymbol{\alpha}}}\left[\left(\dfrac{\phi_j\left(\mathbf{X}\right)f_j\left(\mathbf{X}\right)-\beta_j g_{\boldsymbol{\lambda}_j}\left(\mathbf{X}\right)}{g_{\boldsymbol{\alpha}}\left(\mathbf{X}\right)}\right)^2\right] - \mathbb{E}_{g_{\boldsymbol{\alpha}}}\left(\dfrac{\phi_j\left(\mathbf{X}\right)f_j\left(\mathbf{X}\right)-\beta_j g_{\boldsymbol{\lambda}_j}\left(\mathbf{X}\right)}{g_{\boldsymbol{\alpha}}\left(\mathbf{X}\right)}\right)^2\\
    &= \mathbb{E}_{g_{\boldsymbol{\alpha}}}\left[\dfrac{\left(\phi_j\left(\mathbf{X}\right)f_j\left(\mathbf{X}\right)-\beta_j g_{\boldsymbol{\lambda}_j}\left(\mathbf{X}\right)\right)^2}{g_{\boldsymbol{\alpha}}\left(\mathbf{X}\right)^2}\right] - \underbrace{\mathbb{E}_{f_j}\left(\phi_j\left(\mathbf{X}\right)-\dfrac{\beta_j g_{\boldsymbol{\lambda}_j}\left(\mathbf{X}\right)}{f_j\left(\mathbf{X}\right)}\right)^2}_{= c_j \text{ independent of } \boldsymbol{\alpha}}\\
    &= \mathbb{E}_{h}\left[\dfrac{\left(\phi_j\left(\mathbf{X}\right)f_j\left(\mathbf{X}\right)-\beta_j g_{\boldsymbol{\lambda}_j}\left(\mathbf{X}\right)\right)^2}{g_{\boldsymbol{\alpha}}\left(\mathbf{X}\right)h\left(\mathbf{X}\right)}\right] - c_j.
\end{align*} Therefore, we have: \begin{align*}
    \sum_{j=1}^J w_j\mathbb{V}&_{g_{\boldsymbol{\alpha}}}\left(\dfrac{\phi_j\left(\mathbf{X}\right)f_j\left(\mathbf{X}\right)-\beta_j g_{\boldsymbol{\lambda}_j}\left(\mathbf{X}\right)}{g_{\boldsymbol{\alpha}}\left(\mathbf{X}\right)}\right) \\
    &=\sum_{j=1}^J w_j\left(\mathbb{E}_{h}\left[\dfrac{\left(\phi_j\left(\mathbf{X}\right)f_j\left(\mathbf{X}\right)-\beta_j g_{\boldsymbol{\lambda}_j}\left(\mathbf{X}\right)\right)^2}{g_{\boldsymbol{\alpha}}\left(\mathbf{X}\right)h\left(\mathbf{X}\right)}\right] - c_j\right)\\
    &= \sum_{j=1}^J w_j\mathbb{E}_{h}\left[\dfrac{\left(\phi_j\left(\mathbf{X}\right)f_j\left(\mathbf{X}\right)-\beta_j g_{\boldsymbol{\lambda}_j}\left(\mathbf{X}\right)\right)^2}{g_{\boldsymbol{\alpha}}\left(\mathbf{X}\right)h\left(\mathbf{X}\right)}\right] - \sum_{j=1}^J w_jc_j\\
    &= \mathbb{E}_{h}\left[\dfrac{\sum_{j=1}^Jw_j\left(\phi_j\left(\mathbf{X}\right)f_j\left(\mathbf{X}\right)-\beta_j g_{\boldsymbol{\lambda}_j}\left(\mathbf{X}\right)\right)^2}{g_{\boldsymbol{\alpha}}\left(\mathbf{X}\right)h\left(\mathbf{X}\right)}\right] - \sum_{j=1}^J w_jc_j.
\end{align*} Since the term $\sum_{j=1}^J w_jc_j$ does not depend on the sequence $\boldsymbol{\alpha}$, minimizing $\sum_{j=1}^J w_j\mathbb{V}_{g_{\boldsymbol{\alpha}}}\left(\dfrac{\phi_j\left(\mathbf{X}\right)f_j\left(\mathbf{X}\right)-\beta_j g_{\boldsymbol{\lambda}_j}\left(\mathbf{X}\right)}{g_{\boldsymbol{\alpha}}\left(\mathbf{X}\right)}\right)$ w.r.t. $\boldsymbol{\alpha}$ is then equivalent to minimize $\mathbb{E}_{h}\left[\dfrac{\sum_{j=1}^Jw_j\left(\phi_j\left(\mathbf{X}\right)f_j\left(\mathbf{X}\right)-\beta_j g_{\boldsymbol{\lambda}_j}\left(\mathbf{X}\right)\right)^2}{g_{\boldsymbol{\alpha}}\left(\mathbf{X}\right)h\left(\mathbf{X}\right)}\right]$ w.r.t. $\boldsymbol{\alpha}$. As a conclusion, both optimization problems in Equations \eqref{eq:optimiation_alpha_k} and \eqref{eq:optimisation_alpha_k_expectation} are equivalent.

\end{appendices}

\bibliographystyle{unsrt}
\bibliography{mybibfile}

\end{document}